\documentclass[journal]{IEEEtran}

\usepackage[noadjust]{cite}
\usepackage{graphicx}
\usepackage{amsmath,amsfonts,amssymb,amsthm}
\usepackage{mathtools}
\usepackage{url}
\usepackage[cal=boondoxo]{mathalfa}
\DeclareMathAlphabet{\pazocal}{OMS}{zplm}{m}{n}
\usepackage{color}
\usepackage{enumerate}
\usepackage{enumitem}
\usepackage{subcaption}
\usepackage[format=hang]{caption}
\usepackage[hidelinks]{hyperref}
\usepackage{textcomp}

\usepackage[ruled,linesnumbered]{algorithm2e}

\makeatletter
\newcommand{\removelatexerror}{\let\@latex@error\@gobble}
\makeatother

\allowdisplaybreaks 
\newcounter{storealgline}
\IncMargin{.5em}
\newenvironment{alglist}
{\setcounter{storealgline}{\value{AlgoLine}}%
	\setcounter{AlgoLine}{0}%
	\Indp}
{\setcounter{AlgoLine}{\value{storealgline}}}
\DontPrintSemicolon
\makeatletter
\renewcommand{\Indentp}[1]{%
	\advance\leftskip by #1
	\advance\skiptext by -#1
	\advance\skiprule by #1}%
\renewcommand{\Indp}{\algocf@adjustskipindent\Indentp{\algoskipindent}}
\renewcommand{\Indm}{\algocf@adjustskipindent\Indentp{-\algoskipindent}}
\makeatother

\makeatletter
\renewcommand{\@IEEEsectpunct}{\\ \,}
\makeatother

\theoremstyle{theorem}
\newtheorem{thm}{Theorem}

\newtheorem{cor}[thm]{Corollary}
\newtheorem{lem}[thm]{Lemma}

\newtheorem{asmp}{Assumption}

\newtheorem{defn}{Definition}
\theoremstyle{remark}

\newtheorem{rem}{Remark}

\DeclareMathOperator{\col}{col}
\DeclareMathOperator{\cs}{colspan}
\DeclareMathOperator{\rank}{rank}
\DeclareMathOperator{\diag}{diag}
\DeclareMathOperator{\tr}{tr}

\newcommand{\E}{\mathcal{E}}
\newcommand{\Pk}{\mathcal{P}}

\newcommand{\R}{\mathbb{R}}     
\newcommand{\Z}{\mathbb{Z}}

\newcommand{\EV}{\mathbb{E}}
\newcommand{\Pb}{\mathbb{P}}
\newcommand{\Lt}{\mathcal{L}_2}

\newcommand{\Sym}[1]{\mathbb{S}^{#1}}
\newcommand{\I}{\mathcal{I}}

\newcommand{\n}[1]{\mathtt{n}\left(#1\right)}    
\newcommand{\lag}[1]{\mathtt{L}\left(#1\right)}

\newcommand{\B}{\mathcal{B}}
\newcommand{\Hk}{\mathcal{H}}
\newcommand{\F}{\mathcal{F}}
\newcommand{\Var}{\mathcal{S}}

\newcommand{\bint}[2]{{|[#1,#2]}}
\newcommand{\ra}[1]{#1}

\usepackage{xifthen}
\newcommand{\sR}[2][]{\ifthenelse{\isempty{#1}}{\mathbb{R}^{#2}}{\mathbb{R}^{#1\times #2}}}
\makeatletter
\providecommand{\bigsqcap}{%
	\mathop{%
		\mathpalette\@updown\bigsqcup
	}%
}
\newcommand*{\@updown}[2]{%
	\rotatebox[origin=c]{180}{$\m@th#1#2$}%
}
\makeatother
\begin{document}

\title{Data-Driven Probabilistic Finite $\mathcal{L}_2$-Gain Stabilization of Stochastic Linear Systems}

\author{Yitao Yan, Shuangyu Han, Jie Bao and Biao Huang
\thanks{Y. Yan, S. Han and J. Bao are with the School of Chemical Engineering, UNSW Sydney, NSW 2052, Australia. (e-mail:  y.yan@unsw.edu.au; shuangyu.han@student.unsw.edu.au; j.bao@unsw.edu.au).}
\thanks{B. Huang is with the Department of Chemical and Materials Engineering, University of Alberta, 116 St. and 85 Ave., Edmonton, AB, Canada T6G 2R3, and also with Xi’an  Jiaotong–Liverpool  University,  Suzhou  215123,  China. (e-mail: biao.huang@ualberta.ca)}
}
\maketitle

\begin{abstract}
    In process operations, it is desirable to manage the sensitivity of the system output against external disturbance in the form of finite $\Lt$-gain stabilization. This matter is, however, nonsensical for stochastic systems because the stochastic uncertainties in the control input almost always lead to an unbounded $\Lt$ gain from the disturbance to the output. To address this issue, this article develops a novel concept that characterizes the $\Lt$ gain of stochastic systems in a probabilistic way. Combined with a large data set, we formulate a data-driven probabilistic finite $\mathcal{L}_2$-gain stabilization design using noisy trajectory measurements and the disturbance forecast that does not necessarily agree with the actual future disturbance. The design approach consists of a data-driven trajectory estimation algorithm, whose resulting estimation error covariance is nicely integrated into the feasibility conditions for controller synthesis, leading to a convex offline design in the form of linear matrix inequalities. The effectiveness of the proposed design, along with the additional insights provided by the approach, is illustrated via a numerical example.
\end{abstract}


\IEEEpeerreviewmaketitle

\allowdisplaybreaks

\section{Introduction}
\IEEEPARstart{D}{ata-driven} analysis and control of dynamical systems has received wide attention in the past two decades. Most notably, the introduction of the fundamental lemma in the seminal work by Willems et. al. \cite{Willems:2005} provided a parameterization of the system behavior directly based on one of its persistently exciting trajectories. This lemma has been studied extensively in many aspects, including its equivalent representations \cite{Markovsky:2021,Markovsky:2023}, its extension to the use of multiple trajectories \cite{vanWaarde:2020}, the relaxation of persistent excitation through the informativity framework \cite{vanWaarde:2020a}, the handling of noise in measurements \cite{Alpago:2020,Yan:2024a}, and the representation for stochastic systems using polynomial chaos expansion \cite{Pan:2022,Faulwasser:2023}. Based on the fundamental lemma, various data-driven control approaches have been developed, including quadratic tracking and regulation \cite{Markovsky:2008,daSilva:2018,Yan:2024b,liu2024learning}, Lyapunov and dissipativity design \cite{dePersis:2019,Yan:2025}, control algorithms based on the data-enabled predictive control (DeePC) structure \cite{Coulson:2019,Berberich:2020a}, informativity-based control design approaches \cite{vanWaarde:2025}, distributed control of large-scale systems \cite{Allibhoy:2021,Yan:2024}, control of certain classes of nonlinear systems \cite{Berberich:2022,Alsalti:2023}, and variance control of stochastic systems based on the stochastic fundamental lemma \cite{Pan:2025,Pan:2025a}.

In the control design of dynamical systems, an important aspect of consideration is the regulation of the sensitivity of the output of the controlled system to exogenous disturbances, which is often realized by bounding the $\Lt$ gain of the controlled system. This problem has been studied extensively for deterministic systems, both in the model-based (e.g., \cite{Boyd:1994} and references therein) and in the data-driven context (e.g., \cite{Seuret:2023,Nguyen:2024}). However, characterizing such a bound in practice is difficult because of various sources of inherent stochastic uncertainties. On the one hand, predictive control algorithms often incorporate forecasts in the possible future disturbance to make informed decisions (e.g., the power generated by solar panels based on the weather forecast). While these predictions can be based on historical data, providing good indications of the disturbance trend (as reflected by a mean prediction), the actual future disturbance is most likely a random value around the prediction. On the other hand, when implementing a control algorithm, the actual control input to the system can be uncertain, e.g., the actual flowrate of the fine powder fed through a small orifice may follow a certain probabilistic distribution due to dispersion. In this case, the actual manipulated variable is the expectation of the flowrate rather than the flowrate itself. These two types of uncertainties affect the system dynamics in different ways. Meanwhile, trajectories are typically measured with noise, which is external to the system and masks the inherent uncertainties in the manipulated variable and exogenous disturbance. A potential strategy to manage the impact of disturbances involves developing a certainty-equivalence controller that ignores uncertainties in predictions while incorporating regularization within the algorithm to enhance robustness \cite{Dorfler:2023}. This approach yields a simple control algorithm with good nominal performance, but the certainty-equivalence controller alone is sensitive to prediction errors and measurement noise, and robustness primarily relies on the tuning of the regularization parameters, leading to an empirical description of robust performance. A more well-studied line of research focuses on the worst-case scenario either among all possible admissible disturbance trajectories (e.g., \cite{Yan:2024,Yan:2025}) or within an assumed bounded region (e.g., \cite{Bisoffi:2022}). However, due to the uncertainty in the control input, this rationale is nonsensical in the stochastic setting in general because $\Lt$ gain, as an indicator of the worst-case gain among all admissible trajectories of the system, is almost always unbounded and cannot be regulated.

The above discussions suggest the need for a control approach that is able to systematically characterize the $\Lt$-like performance of stochastic systems in a meaningful way. In this paper, we introduce the novel concept of \emph{probabilistic $\Lt$ gain} and develop a data-driven probabilistic finite $\Lt$-gain stabilization control algorithm to regulate the sensitivity of the system output against disturbance in the presence of control input uncertainty, disturbance prediction error and measurement noise. The controller has access to the noisy measurements of the past trajectory and the disturbance forecast, but the actual disturbance has a random deviation from the predicted value, and the actual control action has a random uncertainty. After implementing the control action, all variables are measurable but with noise. The proposed design contains a filtering algorithm, and the filtering error covariance is an integral part of the control design. These two integrated components give a convex design approach that leads to a characterization of the probability of success for any given performance level (i.e., any prescribed $\Lt$ gain bound). The design reveals a trade-off between the performance of the mean trajectory and the robustness against uncertainties. This trade-off can be formulated purposefully, maximizing the probability of success given any performance requirement in the form of a desired $\Lt$ gain bound.

The remainder of this paper is organized as follows. Section \ref{sec:preliminaries} introduces relevant background information and formulates the problem by defining finite $\Lt$-gain stability in a probabilistic sense. Section \ref{sec:L2Stabilization} presents the main results of this paper, including an estimation algorithm, the solution to the finite $\Lt$-gain stabilization problem, and special treatment when the disturbance has constant mean. A numerical example is presented in Section \ref{sec:example} to demonstrate the effectiveness of the proposed approach. Section \ref{sec:conclusion} concludes this paper and states possible future directions.

\textbf{Notation.} We use the conventional notations $\R$, $\R^\mathsf{n}$, $\R^{\mathsf{m}\times \mathsf{n}}$, $\Z$, $\Z_{\geq0}$, etc. A set with generic variable $w$ is denoted as $\mathbb{W}$, whose dimension is denoted by $\mathsf{w}$. The set of all $\mathsf{n}\times\mathsf{n}$ symmetric matrices is denoted by $\mathbb{S}^\mathsf{n}$. An $\mathsf{n}\times \mathsf{n}$ identity matrix and an $\mathsf{m}\times \mathsf{n}$ zero matrix are denoted by $I_\mathsf{n}$ and $0_{\mathsf{m}\times \mathsf{n}}$, respectively, and the subscripts are dropped when they are clear from context. The Moore-Penrose inverse of a matrix $A$ is denoted by $A^\dagger$ and we define $A_\perp\coloneqq I-AA^\dagger$, $A^\perp\coloneqq I-A^\dagger A$. For matrices $A_1, A_2,\ldots, A_N$ of compatible dimensions, $\col(A_1, A_2,\ldots, A_N)$ represents vertical stacking of them and $\diag(A_1, A_2,\ldots, A_N)$ represents a block diagonal matrix with these matrices on the main diagonal. For a vector $w\in \sR{\mathsf{w}}$, and a symmetric matrix $M\in\Sym{\mathsf{w}}$, denote $\|w\|_M^2\coloneqq w^\top Mw$. The subscript is dropped if $M=I_\mathsf{w}$.

\section{Preliminaries and Problem Formulation}\label{sec:preliminaries}
\subsection{Background on Behavioral Systems Theory}\label{subsec:background}
We begin by introducing concepts in the behavioral framework relevant to this paper for deterministic dynamical systems, followed by their extensions to stochastic systems. The readers are referred to \cite{Willems:1991,Polderman:1998,Willems:2005} for detailed theory on deterministic systems and \cite{Willems:2013,Baggio:2017} on stochastic systems. 

This paper primarily focuses on discrete-time linear time-invariant (LTI) dynamical systems, which can be defined as a triple $\Sigma=(\mathbb{T},\mathbb{W},\B)$, where $\mathbb{T}\subset\mathbb{Z}_{\geq 0}$ is the time axis, $\mathbb{W}$ is a vector space called the signal space, and $\B$, the system behavior, is a complete linear subspace of $\mathbb{W}^\mathbb{T}$ that satisfies $\sigma\B\subset\B$, where $\sigma$ is a shifting operator such that $\sigma w_k=w_{k+1}$ \cite{Willems:1991}. The generic variable of the behavior, $w$, is called the manifest variable of the system. Every LTI system admits a kernel representation of the form
\begin{equation}\label{eq:deterministicKernel}
    R(\sigma,\sigma^{-1})w_k=0,
\end{equation}
where $R(\sigma,\sigma^{-1})$ is a polynomial matrix in the variables $\sigma$ and $\sigma^{-1}$. Suppose that the highest orders for $\sigma$ and $\sigma^{-1}$ in such a representation are $L^+$ and $L^-$, then the smallest possible value of $L\coloneqq L^++L^-$ is an integer invariant of $\B$ represented by \eqref{eq:deterministicKernel} called the lag, which is denoted by $\lag{\B}$. This finite lag allows a trajectory $\tilde{w}\in\B$ to be obtained via repeated weaving of segments from the $(L+1)$-step behavior  \cite{Markovsky:2005}, denoted as
\begin{equation}
    \B_{L+1}=\left\{w\mid\exists w'\in\B, k\in\mathbb{T} \ \text{s.t. } w=w'_\bint{0}{L^++L^-}\right\},
\end{equation}
where $\tilde{w}_\bint{0}{L^++L^-}$ denotes the trajectory segment of $\tilde{w}$ in the interval $[0,L^++L^-]$. In particular, for any trajectories $\tilde{w}^1,\tilde{w}^2\in\B_{L+1}$ that satisfy $L^-\geq\lag{\B}$, if $\tilde{w}_{\bint{L^++1}{L^++L^-}}^1=\tilde{w}_{\bint{0}{L^--1}}^2$, then
\begin{equation}
    \tilde{w}_{\bint{0}{L^+}}^1\wedge\tilde{w}^{2}\in\B_{L+L^++2},
\end{equation}
where $\wedge$ denotes the concatenation of trajectories.The ability to weave segments reduces the analysis of the system behavior $\B$ (a space of functions) to that of $\B_{L+1}$ (a linear subspace of $\sR{(L+1)\mathsf{w}}$). It also makes receding horizon analysis and control design possible.

Let $w=(w_1,w_2)$ be a partition of $w$. If all elements in $w_2$ are free (i.e., for any $w_2$, there exists $w_1$ such that $(w_1,w_2)\in\B$) while none of the elements in $w_1$ are free, then $(w_1,w_2)$ is an input/output partition of $w$, with $w_2$ being the input. The dimension of $w_2$ is the input cardinality of $\B$. A system behavior can also be described with the aid of latent variable $\ell$, leading to the latent variable dynamical system $\Sigma^{full}=(\mathbb{T},\mathbb{W},\mathbb{L},\B^{full})$, where $\B^{full}$ is the full behavior. The system manifest behavior $\B$ can be obtained from the full behavior as $\B=\{w\mid\exists \ell, (w,\ell)\in\B^{full}\}$. The latent variable $\ell$ is said to have the property of state if, for two trajectories $(w^1,\ell^1),(w^2,\ell^2)\in\B^{full}$, that $\ell_{k}^1=\ell_{k}^2$ implies $(\tilde{w}^1_\bint{0}{k-1}\wedge\tilde{w}^2_\bint{k}{\infty},\tilde{\ell}^1_\bint{0}{k-1}\wedge\tilde{\ell}^2_\bint{k}{\infty})\in\B^{full}$. All LTI systems admit state space representations of the form $(\sigma E+F)\ell+Gw=0$, and the smallest dimension of $\ell$ among all such representations is the state cardinality, denoted as $\n{\B}$.

Let $\tilde{w}_\bint{0}{T}$ be a trajectory in $\B_{T+1}$. A Hankel matrix of depth $L+1$ can be constructed using this trajectory as
\begin{equation}
    \Hk_{L+1}(\tilde{w})=\begin{bmatrix}
        \tilde{w}_\bint{0}{L} & \tilde{w}_\bint{1}{L+1} & \cdots & \tilde{w}_\bint{T-L}{T}
    \end{bmatrix}.
\end{equation}
If $\rank(\Hk_{L+1}(\tilde{w}))=(L+1)\mathsf{w}$, then $\tilde{w}_\bint{0}{T}$ is said to be persistently exciting of order $L+1$. If $w=(y,u)$ is an input/output partition and $\B$ is controllable (i.e., for $w^1,w^2\in\B$, there exists $w\in\B$ and $k_1,k_2\in\mathbb{T}$, $k_1\leq k_2$ such that, through the manipulation of $u$, $w_\bint{0}{k_1}=w^1_\bint{0}{k_1}$ and $w_\bint{k_2}{\infty}=w^2_\bint{k_2}{\infty}$), then a representation of $\B_{L+1}$ can be constructed using one of its measured trajectories using the \emph{fundamental lemma} \cite{Willems:2005}. Specifically, given a $(T+1)$-step measured trajectory $\tilde{w}_\bint{0}{T}$, if $\tilde{u}_\bint{0}{T}$ is persistently exciting of order $L+\n{\B}+1$, then $\cs(\Hk_{L+1}(\tilde{w}))=\B_{L+1}$. We emphasize that, in the case when the input contains both control input and disturbance, controllability is defined under the assumption that both of them as manipulable.

The fundamental lemma is only one of the possible options of behavior parameterization (see \cite{Markovsky:2023} for a detailed discussion of other possible representations). More generally, when $L\geq\lag{\B}$, the dimension of $\B_{L+1}$ is $(L+1)\mathsf{u}+\n{\B}$. This means that there exists a matrix $\F\in\sR[((L+1)\mathsf{w})]{((L+1)\mathsf{u}+\n{\B})}$, with full column rank, such that $\cs(\F)=\B_{L+1}$. In the context of receding horizon analysis, this means that, for any trajectory segment $\tilde{w}_k\coloneqq\tilde{w}_\bint{k-L^-}{k+L^+}$, there exists $g_k\in\sR{(L+1)\mathsf{u}+\n{\B}}$ such that
\begin{equation}\label{eq:fundLemma}
    \tilde{w}_k=\F g_k.
\end{equation}
Since the function of the vector $g$ is to parameterize $\tilde{w}_k$, we refer to it as the parameterizer. The possibility of having a more general representation is particularly useful when measurement noise is present, in which case the system behavior can be sufficiently accurately approximated using a large set of measured trajectories as follows.
\begin{lem}[Behavior Approximation \cite{Yan:2024a}]\label{lem:behaviorApprox}
    Let $\mathcal{W}=\{\tilde{w}^m_{i\bint{0}{T}}\}_{i=1}^N$ be a set of trajectories measured with noise, i.e., $\tilde{w}^m_{i\bint{0}{T}}=\tilde{w}_{i\bint{0}{T}}+\tilde{n}_{i\bint{0}{T}}$, where $n$ is a zero-mean white noise with covariance $\Var_n$. Denote
    \begin{equation}
        M_N=\frac{1}{N}\sum_{i=1}^N\Hk_{L+1}(\tilde{w}_i^m)\Hk_{L+1}(\tilde{w}_i^m)^\top-(T-L+1)\tilde{\Var}_n,
    \end{equation}
    where $\tilde{\Var}_n=I_{L+1}\otimes\Var_n$ with $\otimes$ being the Kronecker product. Let $U_N$ be the matrix whose columns are the eigenvectors corresponding to the largest $(L+1)\mathsf{u}+\n{\B}$ eigenvalues of $M_N$. If all input trajectories in $\mathcal{W}$ are persistently exciting of order $L+\n{\B}+1$, then
    \begin{equation}
        d(\cs(U_N),\B_{L+1})\overset{p}{\rightarrow}0,
    \end{equation}
    where $d(\cdot,\cdot)$ denotes the chordal distance between two linear subspaces and $\overset{p}{\rightarrow}$ means convergence in probability. In other words, choosing $U_N$ as $\F$ in \eqref{eq:fundLemma} is sufficiently accurate when $N$ is large. 
\end{lem}

In the case where $L^+=0$, it has been shown in \cite{Yan:2025} that the parameterizer $g$ has the state property and is an observable state variable of the system behavior. An observable state map can also be computed from \eqref{eq:fundLemma} as
\begin{equation}\label{eq:stateMap}
    g_k=\F^\dagger\tilde{w}_k.
\end{equation}

In the stochastic setting, the definition of dynamical system is extended to a quadruple $\Sigma=(\mathbb{T}, \mathbb{W}, \pazocal{E}, \mathbb{P})$, where $\mathbb{T}$ is the time axis, $\mathbb{W}$ is the signal space, $\pazocal{E}$ is the event space containing a $\sigma$-algebra of subsets of $\mathbb{W}^\mathbb{T}$, and $\mathbb{P}:\pazocal{E}\rightarrow [0,1]$ is the probability measure \cite{Baggio:2017}. The deterministic system can be recovered from this definition with $\pazocal{E}=\{\varnothing,\mathbb{W}^\mathbb{T},\B,\mathbb{W}^\mathbb{T}\setminus\B\}$ and $\mathbb{P}\{\B\}=1$. A stochastic system is LTI if there exists a deterministic LTI behavior $\B^{\mathrm{LTI}}$, called the fiber, such that $\pazocal{E}$ is the Borel $\sigma$-algebra generated by the open sets of $\mathbb{W}^\mathbb{T}/\B^{\mathrm{LTI}}$. All stochastic LTI systems admit kernel representations of the form
\begin{equation}\label{eq:stochasticKernel}
    R(\sigma,\sigma^{-1})\ra{w}_k=\ra{e}_k,
\end{equation}
where $R(\sigma,\sigma^{-1})w_k=0$ represents $\B^{\mathrm{LTI}}$, and $\ra{e}$ describes a stochastic system $(\mathbb{T}, \mathbb{W}_e, \pazocal{E}_e, \mathbb{P}_e)$ whose event space is the Borel $\sigma$-algebra generated by the open sets of $\mathbb{W}_e^\mathbb{T}$. A notable difference between deterministic LTI systems represented by \eqref{eq:deterministicKernel} and stochastic LTI systems represented by \eqref{eq:stochasticKernel} is that the position of the timestep $k$ in the former depends on the split of ``past'' and ``future'' segments, whereas that in the latter depends on when the uncertain variable $\ra{e}$ is introduced.

\subsection{Problem Formulation}\label{subsec:problemFormulation}
    In this paper, we consider the regulation of disturbance effect for controllable LTI systems in a general setting, in which uncertainties are present in both the manipulated variable and the exogenous disturbance. Despite the ubiquity of input uncertainties in virtually all control systems, analysis and the regulation of its effect on the control performance has largely been overlooked. For example, in linear quadratic gaussian control design, while the effect of input uncertainty on estimation error is dealt with by Kalman filters, its effect on the disturbance regulation control performance is not handled explicitly. This input uncertainty will result in the controlled system being stochastic in nature, which, combined with measurement noise, requires the control goal to be formulated in a probabilistic way. Furthermore, as will be discussed in details in the next section, attenuating the disturbance effect requires the controller to find a balance between known information (predicted expected value) and unknown error in the disturbance, which is why the primary goal of this paper is finite $\Lt$-gain stabilization even though disturbance prediction will be a known information during online implementation.
    
    Let $\Sigma$ be an LTI system whose manifest variable admits the partition $\ra{w}=(\ra{y},\ra{u},\ra{d})$, where $\ra{y}$ is the controlled variable, $\ra{u}$ is the control input and $\ra{d}$ is the exogenous disturbance. At each timestep $k$, disturbance $d_k$ is with a known mean $\EV[\ra{d}_k]$, but its actual value has a random  deviation, i.e., $\ra{d}_k=\EV[\ra{d}_k]+\Delta \ra{d}_k$, where $\Delta \ra{d}$ is an independent and identically distributed (i.i.d) zero mean stochastic process with covariance $\Var_d$. In the rest of this paper, this is denoted as $\ra{d}\in\mathcal{D}(\EV[\ra{d}],\Var_d)$, where
    \begin{equation}
        \mathcal{D}(\EV[\ra{v}],\Var_v)\coloneqq\left\{\ra{v} \middle | \begin{split}&\ra{v}=\EV[\ra{v}]+\Delta \ra{v}, \ \EV[\Delta \ra{v}_j\Delta \ra{v}_k^\top]=\Var_v\delta_{jk},\\
        & \Delta \ra{v}\text{ is an i.i.d stochastic process}\end{split}\right\},
    \end{equation}
    where $\delta_{jk}$ is the Kronecker delta. To facilitate control design, we make a mild assumption on the boundedness of the expected value of disturbance.
    \begin{asmp}\label{asmp:finiteDist}
        $\|\EV[\ra{d}_k]\|<\infty$ for all $k\in\mathbb{T}$.
    \end{asmp}
    This assumption requires the expected disturbance value to be finite for every step, which is almost always the case in practical applications. Furthermore, while the controller at each timestep decides a control action $\bar{u}_k$, the actual control action upon implementation has a random zero mean uncertainty $\Delta \ra{u}_k$ with covariance $\Var_u$, resulting in a stochastic control action $\ra{u}_k=\bar{u}_k+\Delta \ra{u}_k$. This means that the control input is a stochastic process $\ra{u}\in\mathcal{D}(\bar{u},\Var_u)$, where $\bar{u}$ is the manipulable component. We assume that all components of the manifest variable are measurable but with zero-mean measurement noise $\ra{n}_k$ whose covariance is $\Var_n$. As such, the measured variable can be expressed as
    \begin{equation}
        \ra{w}_k^m=\ra{w}_k+\ra{n}_k=\begin{bmatrix}
            \ra{y}_k \\\bar{u}_k+\Delta \ra{u}_k\\ \EV[\ra{d}_k]+\Delta \ra{d}_k
        \end{bmatrix}+\ra{n}_k.
    \end{equation}
    These measurements are used to construct Hankel matrices and obtain representation following Lemma \ref{lem:behaviorApprox}. The input uncertainty $\Delta u$, disturbance prediction error $\Delta d$ and measurement noise $n$ are assumed to be pair-wise independent.

    The aim of this paper is to achieve $\Lt$ stability for a given system described above. In the deterministic setting, this is often achieved by characterizing a worst-case $\Lt$ gain bound for the controlled system. Specifically, the goal is to control the system such that $\sup_{(\tilde{y},\tilde{d})\in\B,\tilde{d}\neq 0}\frac{\|\tilde{y}_\bint{1}{T}\|}{\|\tilde{d}_\bint{1}{T}\|}$ is finite for all $T$. In practice, finite $\Lt$-gain stabilization design is often carried out using the concept of dissipativity \cite{Willems:1972}, i.e., the construction of a storage function $V\geq0$ and a supply rate $s$ (both as functions of $\tilde{w}$, for example) such that, for the controlled system,
    \begin{equation}\label{eq:dissIneq}
        V_k-V_{k-1}\leq s_k,
    \end{equation}
for all $k$. In particular, to achieve an $\Lt$ gain bound of $\gamma$, the supply rate can be chosen as $s_k=-\|y_k\|^2+\gamma^2\|d_k\|^2$ \cite{Hill:1976}.  However, in the settings of this paper, this is not achievable for any finite $\gamma$. To illustrate, since the magnitude of the trajectory $\Delta \tilde{u}$ can be arbitrarily large, so it can be for the magnitude of $\tilde{y}$ in the controlled system. Furthermore, the control action $\bar{u}_k$ is decided based on the past trajectory measured with noise, adding another layer of uncertainty. This means that, in the stochastic setting, the $\Lt$ gain can be arbitrarily large for any disturbance trajectory with finite value, and the deterministic control setup does not apply. However, since $\Delta u$ has a finite covariance, it is almost surely bounded (which can be shown using the Chebyshev inequality). This means that it is possible to characterize the boundedness of $\Lt$ gain in a probabilistic sense, leading to the following definition of probabilistic $\Lt$ gain bound.
\begin{defn}[Probabilistic $\Lt$ Gain Bound]\label{defn:L2StabilityProb}
    Let $\Sigma=(\mathbb{T}, \mathbb{W}, \pazocal{E}, \mathbb{P})$ be a stochastic dynamical system. Partition its manifest variable $w=(y,d)$, in which $d$ is free. Denote
    \begin{equation}
        \Gamma_T=\frac{\|\tilde{y}_\bint{1}{T}\|}{\|\tilde{d}_\bint{1}{T}\|},  \label{eq:truncL2}
    \end{equation}
    and let $E_{T,\gamma}\in\pazocal{E}$ be the event
    \begin{equation}
        E_{T,\gamma}=\{\tilde{w}\mid \Gamma_T\leq \gamma\}.
    \end{equation}
    The stochastic system $\Sigma$ is said to have an $\Lt$ gain bound of $\gamma$ with probability of failure at most $p_T$ if
    \begin{equation}\label{eq:probL2Defn}
		\Pb\{E_{T,\gamma}\}\geq 1-p_T
	\end{equation}
    for all $T\in\mathbb{T}$. $\Sigma$ is said to have an ultimate $\Lt$ gain bound of $\gamma$ with probability of failure at most $p$ if
    \begin{equation}
		\lim_{T\rightarrow\infty}\Pb\{E_{T,\gamma}\}\geq 1-p.
	\end{equation}
\end{defn}
\begin{rem}
    With a slight abuse of notation, we express \eqref{eq:probL2Defn} as
    \begin{equation}
		\Pb\{\Gamma_T\leq\gamma\}\geq1-p
	\end{equation}
    in the rest of this paper for clarity of presentation.
\end{rem}
If $p$ approaches 0 as $\gamma$ approaches infinity, then we can say that the probability of instability is 0. This leads to the following definition of almost sure $\Lt$ stability.
\begin{defn}[Almost Sure $\Lt$ Stability]\label{defn:L2Stability}
	A stochastic dynamical system $\Sigma$ with manifest variable $w=(y,d)$ is almost surely $\Lt$ stable if
	\begin{equation}
		\Pb\{\Gamma_T<\infty\}=1
	\end{equation}
     for all $T\in\mathbb{T}$.
\end{defn}
This definition does not explicitly specify a universal $\Lt$ gain bound because the randomness of the manifest variable could lead to an arbitrarily large $\Lt$ gain, but it requires the bound to be almost surely finite.

\section{Probabilistic Finite $\mathcal{L}_2$-Gain Stabilization Design}\label{sec:L2Stabilization}
    As illustrated in Section \ref{subsec:background}, the control design can be carried out in the finite interval $[k-L^-,k+L^+]$ for any $L^+ \ge 0$, provided that $L^-\geq\lag{\B}$. For clarity of illustration, all subsequent developments assume that $L^+=0$, i.e., the design is carried out in the interval $[k-L,k]$, and further predictions can be carried out analogously via repeated trajectory weaving. In such a case, $L=L^-$, $\tilde{w}_k=\tilde{w}_\bint{k-L}{k}$, the kernel representation of the uncontrolled system is of the form
    \begin{equation}\label{eq:kernelAllPast}
        R(\sigma^{-1})w_k=0,
    \end{equation}
    and the vector $g$ in \eqref{eq:fundLemma} has the state property. This setup allows the control design to be carried out by controlling the outcomes of $g_k$. However, due to the presence of uncertainties and measurement noise, the exact value of the parameterizer $g_k$ cannot be obtained. In this section, we begin by presenting an optimal filtering approach for parameterizer estimation, which is structurally similar to the Kalman filter design. In addition to converging to the true value of the parameterizer optimally, the resulting estimation error covariance turns out to be also a required component in control design. Using the error covariance, as well as $\Var_u$, $\Var_d$ and $\Var_n$, a set of linear matrix inequalities (LMIs) is developed to assess the feasibility of designing a controller to achieve probabilistic finite $\Lt$-gain stability.
    
    \subsection{Parameterizer Dynamics Estimation}\label{subsec:estimation}
    Given a measured trajectory segment $\tilde{w}_{k-1}$ (which can be parameterized by $g_{k-1}$ using \eqref{eq:fundLemma}) and $d_k$, all possible solutions of $g_k$, can be obtained by solving
    \begin{equation}
        \begin{bmatrix}
            \F_{wp}\\ \F_{dk}
        \end{bmatrix}g_k=\begin{bmatrix}
            \tilde{w}_\bint{k-L}{k-1}\\d_k
        \end{bmatrix}=\begin{bmatrix}
				\Pi_p\F g_{k-1} \\ d_k
			\end{bmatrix},
    \end{equation}
     where $\Pi_p$ is a selection matrix such that $\Pi_p\tilde{w}_k=\tilde{w}_\bint{k-L+1}{k}$. This leads to the representation of $g_k$ as
	\begin{equation}\label{eq:gIter}
		\begin{split}
			g_k&=\begin{bmatrix}
				\F_{wp}\\ \F_{dk}
			\end{bmatrix}^\dagger\begin{bmatrix}
				\Pi_p\F g_{k-1} \\ d_k
			\end{bmatrix}+\begin{bmatrix}
				\F_{wp}\\ \F_{dk}
			\end{bmatrix}^\perp g^{null}_k\\
			&\eqqcolon\F_pg_{k-1}+\F_fd_k+\F_zz_k,
		\end{split}
	\end{equation}
   where $\F_z$ is an orthonormal matrix such that $\cs(\F_z)=\cs\left(\begin{bmatrix}
				\F_{wp}\\ \F_{dk}
			\end{bmatrix}^\perp\right)$. The arbitrary vector $z_k$ can be seen as a virtual manipulated variable that ``controls'' $g$. 
    Since the ultimate design goal is the finite $\Lt$-gain stabilization of the actual stochastic variable $y$ against $d$, control design should focus on their actual values. However, the uncertainties in $u$ and $d$, as well as the measurement noise, make the true values of $g_{k-1}$ and $d_k$, hence that of $y_k$, not available. Furthermore, as illustrated in Section \ref{subsec:problemFormulation}, the actual control input to the system $u_k$ is different from the implementable manipulated variable $\bar{u}_k$. As such, it is not possible to design an implementable control action based on \eqref{eq:gIter} directly. The first step is therefore to decompose the dynamics of the parameterizer into the ``deterministic'' part, which can be used to construct control actions, and the ``uncertain'' part, which the control design needs to compensate for. Since the parameterizer is an observable state of the behavior, we propose to estimate its true value by developing an optimal estimation procedure, similar to a Kalman filter. 

    Since \eqref{eq:gIter} represents the true parameterizer dynamics, the prior estimation based on the information up to the $(k-1)$-th step and the known value of $\EV[d_k]$ can be constructed as
    \begin{equation}\label{eq:prior}
        \hat{g}_{k|k-1}=\F_p\hat{g}_{k-1|k-1}+\F_f\EV[d_k]+\F_z\hat{z}_k,
    \end{equation}
with $\hat{g}_{k-1|k-1}$ the posterior update in the previous horizon (whose construction is presented later in this section) and $\hat{z}_k$ the virtual variable to be manipulated based on the available information. Note that $\hat{z}_k$ is different from $z_k$ in \eqref{eq:gIter} because they are decided based on different information. The associated dynamics for the prior error $e_{k|k-1}\coloneqq g_k-\hat{g}_{k|k-1}$ is then
    \begin{equation}\label{eq:priorError}
        e_{k|k-1}=\F_pe_{k-1|k-1}+\F_f\Delta d_k+\F_z\Delta z_k,
    \end{equation}
where $\Delta z_k=z_k-\hat{z}_k$. The estimated parameterizer in \eqref{eq:prior} leads to the prior estimation of $\tilde{w}_k$, based on which the to-be-implemented manipulated variable $\bar{u}_k$ is obtained as
    \begin{equation}\label{eq:controlAction}
    	\bar{u}_k=\Pi_u\F\hat{g}_{k|k-1},
    \end{equation}
    with $\Pi_u$ a selection matrix such that $u_k=\Pi_u\tilde{w}_k$. 
    
    Since both $z$ and $\hat{z}$ are virtual variables, the statistical properties of $\Delta z$ need to be constructed. This leads to the representation of the prior estimation as in the following lemma.
    \begin{lem}\label{lem:paramDecomp}
    The dynamics of the prior estimation error can be represented as
    \begin{equation}\label{eq:priorTransformed}
    e_{k|k-1}=\E_pe_{k-1|k-1}+\E_f\Delta d_k+\E_u\Delta u_k,
    \end{equation}
    where 
    \begin{equation}\label{eq:errorCoeff}
    \begin{split}
        \E_p&=[I-\F_z(\Pi_u\F\F_z)^{-1}\Pi_u\F]\F_p,\\
        \E_f&=[I-\F_z(\Pi_u\F\F_z)^{-1}\Pi_u\F]\F_f,\\
        \E_u&=\F_z(\Pi_u\F\F_z)^{-1}.
        \end{split}
    \end{equation}
    \end{lem}
    \begin{proof}
    By the construction in \eqref{eq:gIter}, $z_k$ does not contribute to the parameterization of $\tilde{w}_\bint{k-L}{k-1}$ or $d_k$. Since $\dim(\cs(\begin{bmatrix}
        \F_p & \F_f
    \end{bmatrix}))=L\mathsf{u}+(L+1)\mathsf{d}+\n{\B}$ while $\dim(\cs(\F))=(L+1)(\mathsf{u}+\mathsf{d})+\n{\B}$, we have that the dimension of $z_k$ is the same as $\dim(\cs(\F_z))=\mathsf{u}$. As shown by \cite[Lemma 1]{Markovsky:2008}, the combination of $\tilde{w}_\bint{k-L}{k-1}$, $d_k$ and $u_k$ uniquely specifies $y_k$. It then follows that $z_k$ can be completely specified once $u_k$ is determined. This is also true for $\Delta z_k$ and $\Delta u_k$ because their dynamics are the same, as shown by \eqref{eq:priorError}. 

    Within each horizon, $\bar{u}_k$ can be obtained from \eqref{eq:controlAction}, while the actual $u_k$ (with input uncertainty) is parameterized in a similar way with $\hat{g}_{k|k-1}$ replace by $g_k$, which is obtained from \eqref{eq:gIter} using the actual values of $g_{k-1}$, $d_k$ and $z_k$. Using \eqref{eq:priorError}
    \begin{equation}
    \begin{split}
        \Delta u_k&=u_k-\bar{u}_k\\
                  &=\Pi_u\F e_{k|k-1}\\
                  &=\Pi_u\F(\F_pe_{k-1|k-1}+\F_f\Delta d_k+\F_z\Delta z_k)
    \end{split}
    \end{equation}
    As discussed above, $\dim(\cs(\F_z))=\mathsf{u}$, and hence $\F_z\in\sR[\mathsf{g}]{\mathsf{u}}$. Due to the requirement of persistent excitation, $\Pi_u\F$ is of full row rank, which means that $\Pi_u\F\F_z$ is invertible. As such, $z_k$ is of the form
    \begin{equation}\label{eq:zDecomp}
        \Delta z_k=(\Pi_u\F\F_z)^{-1}(\Delta u_k-\Pi_u\F\F_pe_{k-1|k-1}-\Pi_u\F\F_f\Delta d_k).
    \end{equation}
    Substituting into \eqref{eq:priorError} leads to the representation of $e_{k|k-1}$ given by \eqref{eq:priorTransformed}.
    \end{proof}

    The posterior update is structurally similar to that in Kalman filter design, except that the update is based on measured manifest variable $w_k^m$ instead of only measured output, i.e.,
    \begin{equation}\label{eq:posterior}
    \begin{split}
        \hat{g}_{k|k}&=\hat{g}_{k|k-1}+\mathcal{K}_k(w_k^m-\hat{w}_{k|k-1})\\
        &=\hat{g}_{k|k-1}+\mathcal{K}_k\Pi_f\F e_{k|k-1}+\mathcal{K}_kn_k,
    \end{split}
    \end{equation}
where $\Pi_f$ selects $w_k$ from $\tilde{w}_k$. This leads to the posterior error $e_{k|k}=g_k-\hat{g}_{k|k}$ as
    \begin{equation}
        e_{k|k}=(I-\mathcal{K}_k\Pi_f\F) e_{k|k-1}-\mathcal{K}_kn_k.
    \end{equation}
    
Since $e_{k-1|k-1}$, $\Delta d_k$ and $\Delta u_k$ are pair-wise independent, the prior error covariance $\Pk_{k|k-1}=\EV[e_{k|k-1}e_{k|k-1}^\top]$ can be written as
    \begin{equation}\label{eq:KFIterStart}
        \Pk_{k|k-1}=\E_p\Pk_{k-1|k-1}\E_p^\top+\E_f\Var_d\E_f^\top+\E_u\Var_u\E_u^\top.
    \end{equation}
Here, $\Pk_{k|k}=\EV[e_{k|k}e_{k|k}^\top]$ is the posterior error covariance, which can be computed as
    \begin{equation}
        \Pk_{k|k}=(I-\mathcal{K}_k\Pi_f\F)\Pk_{k|k-1}(I-\mathcal{K}_k\Pi_f\F)^\top+\mathcal{K}_k\Var_n\mathcal{K}_k^\top.
    \end{equation}
The minimum of $\tr(\Pk_{k|k})$ can then be found, and the corresponding value of $\mathcal{K}_k$ can be solved as
    \begin{equation}\label{eq:KFGain}
        \mathcal{K}_k=\Pk_{k|k-1}\F^\top\Pi_f^\top(\Pi_f\F\Pk_{k|k-1}\F^\top\Pi_f^\top+\Var_n)^{-1},
    \end{equation}
leading to the optimal covariance
    \begin{equation}\label{eq:KFIterEnd}
        \Pk_{k|k}=(I-\mathcal{K}_k\Pi_f\F)\Pk_{k|k-1}.
    \end{equation}
Since $g$ is an observable state variable of the system whose behavior is controllable, iterating \eqref{eq:KFIterStart} -- \eqref{eq:KFIterEnd} results in $\Pk_{k|k}$ converging exponentially to a unique solution $\Pk$ \cite{Anderson:2005}, which is the solution to the following Algebraic Riccati Equation (ARE):
    \begin{equation}\label{eq:ARE}
        \Pk=\E_p\Pk\E_p^\top-\mathcal{N}+\mathcal{Q},
    \end{equation}
where
    \begin{equation}\label{eq:ARETerms}
    \begin{split}
    \mathcal{Q}&=\E_f\Var_d\E_f^\top+\E_u\Var_u\E_u^\top,\\
   \mathcal{N}&=(\mathcal{A}\Pk\E_p^\top+\mathcal{T})^\top(\mathcal{A}\Pk\mathcal{A}^\top+\mathcal{R})^{-1}(\mathcal{A}\Pk\E_p^\top+\mathcal{T}),\\
\mathcal{A}&=\Pi_f\F\E_p,\\
        \mathcal{R}&=\Pi_f\F(\E_f\Var_d\E_f^\top+\E_u\Var_u\E_u^\top)\F^\top\Pi_f^\top+\Var_n,\\
        \mathcal{T}&=\Pi_f\F(\E_f\Var_d\E_f^\top+\E_u\Var_u\E_u^\top).
        \end{split}
    \end{equation} 
    Note that, different from the conventional steady-state Kalman filter design, where the steady-state solution is for the prior estimate error covariance, the steady state solution $\Pk$ in \eqref{eq:ARE} is for the posterior update error covariance. This steady-state solution will be useful for control design, as will be illustrated in the next section.

\subsection{Probabilistic Finite $\Lt$-Gain Stabilization - the General Case}\label{subsec:dissGeneral}
We now turn our attention to the control design for probabilistic finite $\Lt$-gain stabilization. The probabilistic conditions given in Definitions \ref{defn:L2Stability} and \ref{defn:L2StabilityProb} are difficult to handle directly. Inspired by the stochastic stability design approach, in which the probability of stability can be implied by the expected rate of change of the stochastic Lyapunov function at the $k$th step conditioned on the information available up to the $(k-1)$-th step (denoted as $\mathcal{I}_{k-1}$) \cite{Kushner:1967}, the current approach develops conditions by focusing on the conditional expectation of dissipativity, i.e., the satisfaction of \eqref{eq:dissIneq} conditioned on $\I_{k-1}$. While the construction is similar to the concept of stochastic dissipativity introduced in \cite{Haddad:2022,Lanchares:2023}, we will show that a careful choice of supply rate leads to probabilistic description of the control performance. For LTI systems, it has been shown in \cite{Yan:2025} that the storage function $V$ can be chosen as a quadratic function of the parameterizer $g$, i.e., $V(g)=\|g\|_P^2$, with $P\geq0$. However, since the actual value of $\tilde{w}_{k-1}$, hence that of $g_{k-1}$, is unavailable due to the measurement noise, control design can only be carried out based on the estimated parameterizer $\hat{g}_{k-1|k-1}$ using the decomposition of $g_k$ in Lemma \ref{lem:paramDecomp}. As illustrated in Section \ref{subsec:estimation}, $\hat{g}_{k-1|k-1}$ is obtained based on the dynamic filtering procedure, which has an estimation error covariance of $\Pk_{k-1|k-1}$. Since $\hat{g}_{k|k}$ and $e_{k|k}$ are assumed to be independent for all $k\in\mathbb{W}$, we have
\begin{equation}
    \EV[\|g_k\|_P^2]=\EV[\|\hat{g}_{k|k}+e_{k|k}\|_P^2]=\|\hat{g}_{k|k}\|_P^2+\tr(P\Pk_{k|k}),
\end{equation}
from which it is clear that $\Pk_{k|k}$ will be involved in the control design. However, the time-varying nature of $\Pk_{k|k}$ makes it impractical for control design because the controller needs to be redesigned for every timestep $k$, and this design depends on the initial covariance guess $\Pk_{0|0}$. 

We now present the main result of this paper. It shows that while online updates of the parameterizer estimation follow the dynamic filtering procedure, control design to achieve probabilistic finite $\Lt$-gain stabilization can be carried out using the steady state solution $\Pk$ of the ARE \eqref{eq:ARE}, leading to a time-invariant design. 
\begin{thm}\label{thm:dissConGeneral}
        Let $\Sigma$ be an LTI system represented in the form \eqref{eq:fundLemma} with manifest variable $w=(y,u,d)$, in which $d\in\mathcal{D}(\EV[d],\Var_d)$, where $\EV[d_k]$ satisfies Assumption \ref{asmp:finiteDist}, and $u\in\mathcal{D}(\bar{u},\Var_u)$. Given real numbers $\gamma_1, \gamma_2\geq0$, if there exist matrices $W\in\Sym{\mathsf{g}}$, $X\in\Sym{\mathsf{g}}$, $Y_g\in\sR[\mathsf{u}]{\mathsf{g}}$ and $K_d\in\sR[\mathsf{u}]{\mathsf{d}}$ such that
	\begin{subequations}\label{eq:dissCons}
		\begin{align}
		  &W>0, \label{eq:PD}\\
            &\gamma_2^2\tr(\Var_d)-\tr(X)\geq0,\label{eq:offsetGeneral}\\
            &\begin{bmatrix}
                X & * & *\\
                \mathcal{N}^\frac{1}{2} & W & *\\
                \Pi_y\F\Pk^\frac{1}{2} & 0 & I 
            \end{bmatrix}\geq0,\label{eq:offset_g}\\
			&\begin{bmatrix}
				W&*&*&*\\
				0 & \gamma_1^2I & * & *\\
				\Pi_y\F W & 0 & I & *\\
				\F_pW+\F_zY_g & \F_f+\F_zK_d & 0 & W
			\end{bmatrix}\geq0,\label{eq:dissConGeneral}
		\end{align}
	\end{subequations}
    where $\mathcal{N}$ is defined in \eqref{eq:ARETerms}, and $\Pi_y$ is a selection matrix such that $y_k=\Pi_y\tilde{w}_k$, then there exists a feasible controlled behavior\footnote{Technically speaking, since the controlled system is stochastic, it does not have a controlled behavior, but rather events in the controlled system. For the clarity of illustration, we still use the term ``controlled behavior'' to refer to the set of all possible trajectories in the controlled system.} that satisfies
    \begin{equation}\label{eq:L2Prob}
			\Pb\left\{\lim_{T\rightarrow\infty}\Pb\left\{\Gamma_T\leq\gamma\right\}\geq 1-\frac{1}{\gamma^2}(\rho\gamma_1^2+(1-\rho)\gamma_2^2)\right\}=1,
		\end{equation}
        for any $\gamma\geq\sqrt{\rho\gamma_1^2+(1-\rho)\gamma_2^2}$, with
        \begin{equation}\label{eq:SNRGeneral}
            \rho=\lim_{T\rightarrow\infty}\frac{\frac{1}{T}\sum_{k=1}^T\|\EV[d_k]\|^2}{\tr(\Var_d)+\frac{1}{T}\sum_{k=1}^T\|\EV[d_k]\|^2}.
        \end{equation}
    In such a case, a corresponding solution of the prior parameterizer estimation $\hat{g}_{k|k-1}$ is given by
	\begin{equation}\label{eq:controlConstruct}
		\hat{g}_{k|k-1}=(\F_p+\F_zY_gW^{-1})\hat{g}_{k-1|k-1}+(\F_f+\F_zK_d)\EV[d_k],
	\end{equation}
where $\hat{g}_{k-1|k-1}$ is obtained using \eqref{eq:posterior} -- \eqref{eq:KFIterEnd} with $k\leftarrow k-1$. The control action to be implemented, $\bar{u}_k$, can be obtained using \eqref{eq:controlAction}.
\end{thm}
\begin{proof}
    We begin by explaining the rationale behind the construction of the feasible controlled behavior. Let
\begin{equation}
\begin{split}
    \ra{V}_k&\coloneqq\|\ra{g}_k\|_P^2,\\
    \ra{s}_k&\coloneqq-\|\ra{y}_k\|^2+\gamma_1^2\|\EV[\ra{d}_k]\|^2+\gamma_2^2\|\Delta \ra{d}_k\|^2,
\end{split}
\end{equation}
where $P\geq 0$, and consider the expectation of $\ra{V}_{k-1}-\ra{V}_k+\ra{s}_k$ conditioned on $\I_{k-1}$. Using Lemma \ref{lem:paramDecomp}, we have
    \begin{align}\label{eq:dissRateDecomp}
            &\EV[\ra{V}_{k-1}-\ra{V}_k+\ra{s}_k\mid \I_{k-1}]\nonumber \\
            = \ &\EV[\|\ra{g}_{k-1}\|_P^2-\|\ra{g}_k\|_P^2-\|\ra{y}_k\|^2\nonumber \\
            &\hskip 0.5cm +\gamma_1^2\|\EV[\ra{d}_k]\|^2+\gamma_2^2\|\Delta \ra{d}_k\|^2\mid\I_{k-1}]\nonumber \\
            = \ &\EV[\|g_{k-1}\|_P^2-\|g_k\|_{P+\F^\top\Pi_y^\top\Pi_y\F}^2\nonumber \\
            &\hskip 0.5cm +\gamma_1^2\|\EV[d_k]\|^2+\gamma_2^2\|\Delta d_k\|^2\mid\I_{k-1}]\nonumber \\
            = \ &\EV[\|\hat{g}_{k-1|k-1}+e_{k-1|k-1}\|_{M-\F^\top\Pi_y^\top\Pi_y\F}^2\nonumber \\
            &\hskip 0.5cm -\|\hat{g}_{k|k-1}+e_{k|k-1}\|_M^2\nonumber \\
            &\hskip 0.5cm +\gamma_1^2\|\EV[d_k]\|^2+\gamma_2^2\|\Delta d_k\|^2]\nonumber \\
            = \ & \|\hat{g}_{k-1|k-1}\|_{M-\F^\top\Pi_y^\top\Pi_y\F}^2+\gamma_1^2\|\EV[d_k]\|^2\nonumber \\
            & \hskip 0.5cm -\|\F_p\hat{g}_{k-1|k-1}+\F_f\EV[d_k]+\F_z\hat{z}_k\|_M^2\nonumber \\
            & \hskip 0.5cm +\EV[\|e_{k-1|k-1}\|_{M-\F^\top\Pi_y^\top\Pi_y\F}^2]+\gamma_2^2\EV[\|\Delta d_k\|^2]\nonumber \\
            &\hskip 0.5cm -\EV[\|\E_pe_{k-1|k-1}+\E_f\Delta d_k+\E_u\Delta u_k\|_M^2]\nonumber \\
        = \ &\|\hat{g}_{k-1|k-1}\|_{M-\F^\top\Pi_y^\top\Pi_y\F}^2+\gamma_1^2\|\EV[d_k]\|^2\nonumber \\
        & \hskip 0.5cm -\|\F_p\hat{g}_{k-1|k-1}+\F_f\EV[d_k]+\F_z\hat{z}_k\|_M^2\nonumber \\
        & \hskip 0.5cm +\EV[\|e_{k-1|k-1}\|_{M-\E_p^\top M\E_p-\F^\top\Pi_y^\top\Pi_y\F}^2]\nonumber \\
        & \hskip 0.5cm +\EV[\|\Delta d_k\|_{\gamma_2^2I-\E_f^\top M\E_f}^2]-\EV[\|\Delta u_k\|_{\E_u^\top M\E_u}^2]\nonumber \\
        = \ &\|\hat{g}_{k-1|k-1}\|_{M-\F^\top\Pi_y^\top\Pi_y\F}^2+\gamma_1^2\|\EV[d_k]\|^2\nonumber \\
        & \hskip 0.5cm -\|\F_p\hat{g}_{k-1|k-1}+\F_f\EV[d_k]+\F_z\hat{z}_k\|_M^2\nonumber \\
        & \hskip 0.5cm +\tr((M-\E_p^\top M\E_p-\F^\top\Pi_y^\top\Pi_y\F)\Pk_{k-1|k-1})\nonumber \\
        & \hskip 0.5cm +\tr((\gamma_2^2I-\E_f^\top M\E_f)\Var_d)-\tr((\E_u^\top M\E_u)\Var_u)\nonumber \\
        = \ &\|\hat{g}_{k-1|k-1}\|_{M-\F^\top\Pi_y^\top\Pi_y\F}^2+\gamma_1^2\|\EV[d_k]\|^2\nonumber \\
        & \hskip 0.5cm -\|\F_p\hat{g}_{k-1|k-1}+\F_f\EV[d_k]+\F_z\hat{z}_k\|_M^2\nonumber \\
        &\hskip 0.5cm + \gamma_2^2\tr(\Var_d)-\tr(\F^\top\Pi_y^\top\Pi_y\F\Pk)\nonumber \\
        &\hskip 0.5cm +\tr(M(\Pk-\E_p\Pk\E_p^\top-\E_f\Var_d\E_f^\top-\E_u\Var_u\E_u^\top))\nonumber \\
        &\hskip 0.5cm +\tr((M-\E_p^\top M\E_p-\F^\top\Pi_y^\top\Pi_y\F)(\Pk_{k-1|k-1}-\Pk)),
   \end{align}
    where
    \begin{equation}\label{eq:PMTrans}
        M=P+\F^\top\Pi_y^\top\Pi_y\F.
    \end{equation}
    Since $\Pk$ satisfies the ARE \eqref{eq:ARE}, we have
    \begin{equation}\label{eq:matrixNEquiv}
        \Pk-\E_p\Pk\E_p^\top-\E_f\Var_d\E_f^\top-\E_u\Var_u\E_u^\top=-\mathcal{N}.
    \end{equation}
    Define
    \begin{equation}
        \alpha_k=\tr((M-\E_p^\top M\E_p-\F^\top\Pi_y^\top\Pi_y\F)(\Pk_{k-1|k-1}-\Pk)).
    \end{equation}
    Since $\Pk_{k|k}$ converges to $\Pk$ exponentially, we have
    \begin{equation}
    \begin{split}
        &\sum_{k=1}^\infty\alpha_k\\
        = \ &\tr\left((M-\E_p^\top M\E_p-\F^\top\Pi_y^\top\Pi_y\F)\sum_{k=1}^\infty(\Pk_{k-1|k-1}-\Pk)\right)\\
        < \ &\infty.
    \end{split}
    \end{equation}
   In addition, if $M-\F^\top\Pi_y^\top\Pi_y\F\geq0$ (which is implied by \eqref{eq:dissConGeneral}), then $P\geq0$, hence $V\geq0$. In such a case, if there exists $\hat{z}_k$ (as a function of $\hat{g}_{k-1|k-1}$ and $\EV[d_k]$) such that
    \begin{equation}\label{eq:dissTransform}
        \begin{split}
        &\|\hat{g}_{k-1|k-1}\|_{M-\F^\top\Pi_y^\top\Pi_y\F}^2+\gamma_1^2\|\EV[d_k]\|^2\\
        & \hskip 0.5cm -\|\F_p\hat{g}_{k-1|k-1}+\F_f\EV[d_k]+\F_z\hat{z}_k\|_M^2\\
        &\hskip 1.2cm + \gamma_2^2\tr(\Var_d)-\tr(M\mathcal{N}+\F^\top\Pi_y^\top\Pi_y\F\Pk)\geq0
        \end{split}
    \end{equation}
    for all $\hat{g}_{k-1|k-1}$ and $\EV[d_k]$, then the feasible controlled behavior satisfies
    \begin{equation}
        \EV[V_k\mid\I_{k-1}]-V_{k-1}\leq \EV[s_k\mid\I_{k-1}]-\alpha_k.
    \end{equation}
    Using the law of total expectation, we have
    \begin{equation}
        \EV[V_k]-\EV[V_{k-1}]\leq -\EV[\|y_k\|^2]+\gamma_1^2\|\EV[d_k]\|^2+\gamma_2^2\|\Delta d_k\|^2-\alpha_k.
    \end{equation}
    Summing from 1 to $T$ and rearranging, we have
    \begin{equation}\label{eq:yNormBound}
		\begin{split}
			&\EV\left[\sum_{k=1}^T\|y_k\|^2\right]\\
            \leq \ &\EV[V_0]-\EV[V_N]+\sum_{k=1}^T\gamma_1^2\|\EV[d_k]\|^2+\gamma_2^2\|\Delta d_k\|^2-\alpha_k\\
			\leq \ & \EV[V_0]+T\gamma^2_2\tr(\Var_d)+\sum_{k=1}^T\gamma_1^2\|\EV[d_k]\|^2-\alpha_k.
		\end{split}
	\end{equation}
    Using Markov's inequality,
	\begin{equation}\label{eq:probBound}
		\begin{split}
			&\Pb\left\{\sum_{k=1}^T\|y_k\|^2\geq \gamma^2\sum_{k=1}^T\|d_k\|^2\right\}\\
			\leq \ &\frac{\EV\left[\sum_{k=1}^T\|y_k\|^2\right]}{\gamma^2\sum_{k=1}^T\|d_k\|^2}\\
			\leq \ &\frac{\EV[V_0]+T\gamma^2_2\tr(\Var_d)+\sum_{k=1}^T(\gamma_1^2\|\EV[d_k]\|^2-\alpha_k)}{\gamma^2\sum_{k=1}^T(\|\EV[d_k]\|^2+\|\Delta d_k\|^2+2\EV[d_k]^\top\Delta d_k)}\\
            = \ & \frac{\frac{\EV[V_0]}{T}+\gamma^2_2\tr(\Var_d)+\frac{1}{T}\sum_{k=1}^T(\gamma_1^2\|\EV[d_k]\|^2-\alpha_k)}{\frac{\gamma^2}{T}\sum_{k=1}^T(\|\EV[d_k]\|^2+\|\Delta d_k\|^2+2\EV[d_k]^\top\Delta d_k)}.
		\end{split}
	\end{equation}
    Since both $\EV[V_0]$ and $\sum_{k=1}^\infty\alpha_k$ are finite, we have
    \begin{equation}
        \lim_{T\rightarrow\infty}\frac{\EV[V_0]}{T}=0, \ \lim_{T\rightarrow\infty}\frac{1}{T}\sum_{k=1}^T\alpha_k=0.
    \end{equation}
    Furthermore, since $\|\EV[d_k]\|$ is bounded for every $k$, 
    \begin{equation}
    \begin{split}
        \lim_{T\rightarrow\infty}\frac{1}{T}\sum_{k=1}^T\|\EV[d_k]\|^2&\leq\lim_{T\rightarrow\infty}\frac{1}{T}(T\max_k\|\EV[d_k]\|^2)\\
        &=\max_k\|\EV[d_k]\|^2<\infty.
    \end{split}
    \end{equation}
    This means that the limits of both the numerator and the denominator of the last line of \eqref{eq:probBound} exist, and we have
    \begin{equation}\label{eq:innerProbLim}
        \begin{split}
            &\lim_{T\rightarrow\infty}\Pb\left\{\sum_{k=1}^T\|y_k\|^2\geq \gamma^2\sum_{k=1}^T\|d_k\|^2\right\}\\
            \leq \ & \lim_{T\rightarrow\infty}\frac{\gamma^2_2\tr(\Var_d)+\frac{\gamma_1^2}{T}\sum_{k=1}^T\|\EV[d_k]\|^2}{\frac{\gamma^2}{T}\sum_{k=1}^T(\|\EV[d_k]\|^2+\|\Delta d_k\|^2+2\EV[d_k]^\top\Delta d_k)}.
        \end{split}
    \end{equation}
    As such,
	\begin{align}
			&\Pb\left\{\lim_{T\rightarrow\infty}\Pb\left\{\sum_{k=1}^T\|y_k\|^2\geq \gamma^2\sum_{k=1}^T\|d_k\|^2\right\}\right.\nonumber \\
            &\hskip 3cm \left.\phantom{\sum_{k=1}^T}\leq\frac{1}{\gamma^2}(\rho\gamma_1^2+(1-\rho)\gamma_2^2)\right\}\nonumber \\
			= \ &\Pb\left\{\lim_{T\rightarrow\infty}\Pb\left\{\sum_{k=1}^T\|y_k\|^2\geq \gamma^2\sum_{k=1}^T\|d_k\|^2\right\}\right.\nonumber \\
			&\hskip 1.5cm\left.\leq\lim_{T\rightarrow\infty}\frac{\gamma^2_2\tr(\Var_d)+\frac{\gamma_1^2}{T}\sum_{k=1}^T\|\EV[d_k]\|^2}{\gamma^2(\tr(\Var_d)+\frac{1}{T}\sum_{k=1}^T\|\EV[d_k]\|^2)}\right\}\nonumber \\
			\geq \ &\Pb\left\{\eqref{eq:innerProbLim} \text{ is true}\phantom{\sum_{k=1}^T}\right.\nonumber \\
			&\hskip 1cm\text{ and } \nonumber \\
			&\hskip 0.5cm\lim_{T\rightarrow\infty}\frac{1}{T}\sum_{k=1}^T\|\Delta d_k\|^2=\tr(\Var_d)\nonumber \\
			&\hskip 1cm\text{ and } \nonumber \\
			&\hskip 0.5cm\left.\lim_{T\rightarrow\infty}\frac{1}{T}\sum_{k=1}^T\EV[d_k]^\top\Delta d_k=0\right\}.
	\end{align}
	Since the above three statements are independent and all have probability 1 (by \eqref{eq:innerProbLim} for the first statement and by the strong law of large numbers for the last two), the overall probability is also 1. In other words, if \eqref{eq:dissTransform} can be satisfied for all $\hat{g}_{k-1|k-1}$ and $\EV[d_k]$ by choosing the appropriate value of $\hat{z}_k$, then there exists a feasible controlled behavior that satisfies \eqref{eq:L2Prob}.
    
    The next part of the proof requires the following lemma.
	\begin{lem}\label{lem:QFbehaviorOffset}
		Suppose that $v=(v_1,v_2)$ satisfies the inequality 
		\begin{equation}\label{eq:QFbehaviorOffset}
			\begin{bmatrix}
			    v_1\\ v_2
			\end{bmatrix}^\top\begin{bmatrix}
                Q & S\\ S^\top & -R
            \end{bmatrix} \begin{bmatrix}
			    v_1\\ v_2
			\end{bmatrix}+\begin{bmatrix}
			    \eta\\ \mu
			\end{bmatrix}^\top \begin{bmatrix}
			    v_1\\ v_2
			\end{bmatrix}+\beta\geq0,
		\end{equation}
		where $R\geq0$. The component $v_1$ is free if and only if there exist $K\in\sR[\mathsf{v}_2]{\mathsf{v}_1}$ and $\xi\in\sR{\mathsf{v}_2}$ such that
		\begin{equation}\label{eq:gainQF}
			v_2=Kv_1+\xi
		\end{equation}
		satisfies \eqref{eq:QFbehaviorOffset} for all $v_1$. Furthermore, if $\mu=0$, then $\beta\geq 0$ and $\xi$ can be chosen as 0.
	\end{lem}
	\begin{proof}
        See Appendix \ref{appx:proofQFbehaviorOffset}
	\end{proof}
        
        Using Lemma \ref{lem:QFbehaviorOffset}, \eqref{eq:dissTransform} is of the form \eqref{eq:QFbehaviorOffset} without the linear term, with $v_1=\col(\hat{g}_{k-1|k-1},\EV[d_k])$ and $\beta=\gamma_2^2\tr(\Var_d)-\tr(M\mathcal{N}+\F^\top\Pi_y^\top\Pi_y\F\Pk)$. Therefore, $\hat{g}_{k-1|k-1}$ and $\EV[d_k]$ are free if and only if
        \begin{equation}\label{eq:offsetTotal}
            \gamma_2^2\tr(\Var_d)-\tr(M\mathcal{N}+\F^\top\Pi_y^\top\Pi_y\F\Pk)\geq0,
        \end{equation}
        and there exist $K_g$ and $K_d$ such that
        \begin{equation}\label{eq:zkForm}
            \hat{z}_k=K_g\hat{g}_{k-1|k-1}+K_d\EV[d_k]
        \end{equation}
        satisfies
        \begin{equation}
			\begin{split}
				&\|\hat{g}_{k-1|k-1}\|_{M-\F^\top\Pi_y^\top\Pi_y\F}^2+\gamma_1^2\|\EV[d_k]\|^2\\
				&\hskip0.9cm-\|\F_p\hat{g}_{k-1|k-1}+\F_f\EV[d_k]+\F_z\hat{z}_k\|_M^2\geq0
			\end{split}
			\label{eq:offsetReplace}
		\end{equation}
	for all $\hat{g}_{k-1|k-1}$ and $\EV[d_k]$. Since $\mathcal{N}\geq0$ and $\Pk\geq0$, using the cyclic property of trace, Eq. \eqref{eq:offsetTotal} is equivalent to the existence of a matrix $X$ such that \eqref{eq:offsetGeneral} is satisfied and
    \begin{equation}\label{eq:slack}
        X\geq \mathcal{N}^\frac{1}{2}M\mathcal{N}^\frac{1}{2}+\Pk^\frac{1}{2}\F^\top\Pi_y^\top\Pi_y\F\Pk^\frac{1}{2}.
    \end{equation}
    If $M>0$ (which becomes Condition \eqref{eq:PD}), taking Schur complement with respect to $M$ and defining $W=M^{-1}>0$ leads to Condition \eqref{eq:offset_g}. For \eqref{eq:offsetReplace}, substituting \eqref{eq:zkForm} gives
    \begin{equation}\label{eq:nominalDiss}
			\begin{split}
				&\|\hat{g}_{k-1|k-1}\|_{M-\F^\top\Pi_y^\top\Pi_y\F}^2+\gamma_1^2\|\EV[d_k]\|^2\\
				&\hskip0.2cm-\|(\F_p+\F_zK_g)\hat{g}_{k-1|k-1}+(\F_f+\F_zK_d)\EV[d_k]\|_M^2\geq0,
			\end{split}
		\end{equation}
    which is required to hold for all $\hat{g}_{k-1|k-1}$ and $\EV[d_k]$. This means that
    \begin{equation}
		\begin{split}
		&\begin{bmatrix}
			M&*\\
			0&\gamma_1^2I
		\end{bmatrix}-[*]^\top\begin{bmatrix}
			I&0\\ 0 &M
		\end{bmatrix}
		\begin{bmatrix}
			\Pi_y\F & 0\\
			\F_p+\F_zK_g & \F_f+\F_zK_d
		\end{bmatrix}\geq0.
		\end{split}
    \end{equation}
Using Schur complement, performing congruence transformation using $\diag(W,I,I,I)$, and defining $Y_g=K_gW$ give \eqref{eq:dissConGeneral}.

In summary, if the conditions in \eqref{eq:dissCons} are satisfied, then there exist an $M>0$ and an associated choice of $\hat{z}_k$ in the form \eqref{eq:zkForm} such that \eqref{eq:dissTransform} is satisfied for all $\hat{g}_{k-1|k-1}$ and $\EV[d_k]$. This implies the existence of a feasible controlled behavior satisfying \eqref{eq:L2Prob}. Finally, using the solutions of $W$, $Y_g$ and $K_d$, the value of the $\hat{g}_{k|k-1}$ can be obtained by substituting \eqref{eq:zkForm} into \eqref{eq:prior}, which gives \eqref{eq:controlConstruct}.
\end{proof}

\begin{rem}
    Theorem \ref{thm:dissConGeneral} shows that the proposed approach naturally integrates the estimation algorithm in Section \ref{subsec:estimation} into the control design. This can be more clearly observed from the last equality of \eqref{eq:dissRateDecomp}. Firstly, the exponential convergence of $\Pk_{k|k}$ allows ${\alpha_k}$ to be a convergent series, making it possible for offline control design to be developed using the steady-state covariance $\Pk$, while online estimation still uses the optimal error covariance $\Pk_{k|k}$. Secondly, due to the presence of the term $M\Pk$, \eqref{eq:dissRateDecomp} cannot be converted into a convex condition in general. However, $\Pk$ satisfies the ARE nicely organizes \eqref{eq:dissRateDecomp} into the form \eqref{eq:dissTransform}, which can be converted into a set of LMIs that are readily solvable by existing toolboxes.
\end{rem}

The result of Theorem~\ref{thm:dissConGeneral} shows that the final probabilistic performance is affected by several factors. We now discuss them individually.
\subsubsection{Probabilistic $\Lt$ gain $\gamma$}
The inner probability in \eqref{eq:L2Prob} shows that, if the conditions in \eqref{eq:dissCons} are satisfied, then, for any given gain bound $\gamma$, the probability of violating it is inversely proportional to $\gamma^2$. In fact, a closer observation of \eqref{eq:probBound} shows that, for any value of $N$, the probability of violating $\Lt$ gain bound of $\gamma$ as $\gamma\rightarrow\infty$ is zero. This is also true in the limiting case because both $\gamma_1$ and $\gamma_2$ are finite constants, and $\rho\in[0,1]$. In other words, the satisfaction of \eqref{eq:dissCons} will guarantee almost sure $\Lt$ stability of the controlled system in the sense of Definition \ref{defn:L2Stability}. However, for any fixed value of $\gamma$, there is a non-zero probability of failure, which implies that, for a given trajectory $(\tilde{y},\tilde{d})$ in the controlled system, the worst-case of the ratio $\frac{\|\tilde{y}_\bint{0}{T}\|}{\|\tilde{d}_\bint{0}{T}\|}$ among all $T$ can be arbitrarily large (albeit with arbitrarily low chance). In particular, for small values of $T$, the potentially large values of $\EV[V_0]$ (caused by large value of $g_0$) and $-\alpha_k$ (caused by underestimation of $\Pk$ from $\Pk_{0|0}$) would require a large value of $\gamma$ to bring the probability of failure in \eqref{eq:probBound} to a reasonably small value. These results are not unexpected because $\Lt$ gain is a worst-case condition among all admissible trajectories $(\tilde{y},\tilde{d})$, which could be large initially given the stochastic nature of the system and the measurement noise. Furthermore, the design approach leads to a nested probabilistic result on the $\Lt$ stability of the controlled system. The outer probability is due to the strong law of large numbers, which is well known to converge to 1 rapidly. This means that, for small values of $T$, it is possible that the inner probability cannot be achieved (due to the combined effects of $\Delta d$, $\Delta u$ and the estimation error), but for large values of $T$, the inner probability can almost surely be guaranteed regardless of the distribution of $d_k$. In other words, the satisfaction of \eqref{eq:dissCons} leads to an ultimate probabilistic $\Lt$ gain bound almost surely. As will be demonstrated in the illustrative example in Section \ref{sec:example}, the effect of the outer probability is inconsequential for sufficiently long trajectories. The subsequent discussions will therefore primarily focus on the interpretation of the inner probability.

\subsubsection{Design parameters $\gamma_1$ and $\gamma_2$}
In the general case, these two parameters reflect the performance of two different components in \eqref{eq:dissTransform}. The meaning of $\gamma_1$ can be analyzed from \eqref{eq:controlConstruct}. Taking the expectation on both sides shows that the transition from $\hat{g}_{k-1|k-1}$ to $\hat{g}_{k|k-1}$ is the same as that from $\EV[g_{k-1}]$ to $\EV[g_k]$. Therefore, if \eqref{eq:nominalDiss} is valid for all $\hat{g}_{k-1|k-1}$, then it holds when $\hat{g}_{k-1|k-1}=\EV[g_{k-1}]$. This means that \eqref{eq:nominalDiss} implies the condition 
\begin{equation}\label{eq:EVDiss}
        \|\EV[g_k]\|_P^2-\|\EV[g_{k-1}]\|_P^2\leq -\|\EV[y_k]\|^2+\gamma_1^2\|\EV[d_k]\|^2,
    \end{equation}
in the controlled system. In other words, the $\Lt$ gain of the behavior of the expected manifest variable is bounded by $\gamma_1$. From the perspective of the dynamical system, this gain has another meaning. Suppose that the uncontrolled behavior admits the representation \eqref{eq:kernelAllPast}. Using the state map \eqref{eq:stateMap}, the behavior of the controller constructed from \eqref{eq:controlAction} and \eqref{eq:controlConstruct} can be written as
\begin{equation}
\begin{split}
    u_k&=\Pi_u\F(\F_p+\F_zY_gW^{-1})(\F^\dagger\tilde{w}_{k-1}-e_{k-1|k-1})\\
    &\hskip 1.2cm +\Pi_u\F(\F_f+\F_zK_d)(d_k-\Delta d_k)+\Delta u_k.
\end{split}
\end{equation}
After rearranging, this can be brought to the form
\begin{equation}
    R_c(\sigma^{-1})w_k=M_c\epsilon_k,
\end{equation}
where $\epsilon_k=\col(e_{k-1|k-1},\Delta u_k, \Delta d_k)$. As a result, the controlled system, being the interconnection of the uncontrolled system and the controller, admits the common behavior, i.e., all trajectories satisfying
\begin{equation}
    \begin{bmatrix}
        R(\sigma^{-1})\\ R_c(\sigma^{-1})
    \end{bmatrix}w_k=\begin{bmatrix}
        0\\ M_c
    \end{bmatrix}\epsilon_k.
\end{equation}
This shows that the controlled system is a stochastic LTI system whose representation is of the form \eqref{eq:stochasticKernel}. Since $\EV[\epsilon_k]=0$, taking the expectation on both sides gives the controlled behavior of $\EV[w]$, in which $\EV[y]$ and $\EV[d]$ are the manifest variables and $\EV[u]$ is the latent variable. The $\Lt$ gain on the expected manifest variable  therefore implies the $\Lt$ stability of the fiber of the controlled system.

The meaning of $\gamma_2$ can be explored from \eqref{eq:dissRateDecomp}. Specifically, from the second-last equality of \eqref{eq:dissRateDecomp}, combining with Lemma \ref{lem:QFbehaviorOffset}, it is apparent that achieving the desired performance requires
\begin{equation}
\begin{split}
    &\gamma_2^2\tr(\Var_d)-\tr(M(\F_f\Var_d\F_f^\top+\F_u\Var_u\F_u^\top))\\
    &\hskip 0.5cm+((M-\F_p^\top M\F_p-\F^\top\Pi_y^\top\Pi_y\F)\Pk_{k-1|k-1})\geq0.
    \end{split}
\end{equation}
Rearranging, using \eqref{eq:KFIterStart} and substituting \eqref{eq:PMTrans}, the above inequality can be translated to
\begin{equation}\label{eq:errorDiss}
\begin{split}
    &\EV[\|e_{k|k-1}\|_P^2]-\EV[\|e_{k-1|k-1}\|_P^2]\\
    &\hskip 2cm\leq -\EV[\|\Pi_y\F e_{k|k-1}\|^2]+\gamma_2^2\EV[\|\Delta d_k\|^2],
    \end{split}
\end{equation}
i.e., it reflects the sensitivity of the prior estimation error covariance of $y_k$ against the covariance of $\Delta d_k$. The difference between $\EV[\|e_{k|k-1}\|_P^2]$ and $\EV[\|e_{k-1|k-1}\|_P^2]$ is governed by \eqref{eq:priorTransformed}, which means this difference may be positive or negative, and its magnitude depends on $e_{k-1|k-1}$. By having a small $e_{k-1|k-1}$, both these terms, as well as $\EV[\|\Pi_y\F e_{k|k-1}\|^2]$ (whose sign is negative and therefore contributes negatively to the magnitude of $\gamma_2$), can be maintained at a relatively small magnitude, allowing $\gamma_2$ to have a small value. This is in alignment with our choice of an optimal filtering algorithm, which aims to minimize $\EV[\|e_{k|k}\|^2]$.

While $\gamma_1$ and $\gamma_2$ do not appear in the same inequality, their values are implicitly linked by matrix $M$ in a conflicting way. On the one hand, \eqref{eq:nominalDiss} needs to hold for all $\hat{g}_{k-1|k-1}$ when $\EV[d_k]=0$, leading to the inequality
\begin{equation}
    M-(\F_p+\F_zK_g)^\top M (\F_p+\F_zK_g)-\F^\top\Pi_y^\top\Pi_y\F\geq0,
\end{equation}
which means that the quadratic term of $\hat{g}_{k-1|k-1}$ (the dominating term for any given value of $\EV[d_k]$ in \eqref{eq:nominalDiss}) is non-negative. Having larger eigenvalues in $M$ would lead to a larger positive value for this quadratic term for any given value of $\hat{g}_{k-1|k-1}$, allowing for a smaller value of $\gamma_1$. On the other hand, \eqref{eq:offsetTotal} suggests that a smaller value of $\gamma_2$ can be obtained if the eigenvalues of $M$ are small because reducing the magnitude of $\tr(M\mathcal{N})$ leads to a smaller negative contributing factor. This conflict can be interpreted as a trade-off between the bounds for ``nominal'' performance (represented by $\gamma_1$) and probabilistic ``robust'' performance (represented by $\gamma_2$), which are a pair of conflicting requirements. As such, while their values are given in Theorem \ref{thm:dissConGeneral}, the design can also be carried out by fixing one of them while optimizing the other one.
    
\subsubsection{Weighting $\rho$}
As shown in \eqref{eq:L2Prob}, the probability of the actual gain is a function of the weighted average of the two $\gamma_1^2$ and $\gamma_2^2$, and the weighting is given by the parameter $\rho\in[0,1]$. This weighting reflects the relative magnitude between $\|\EV[d_k]\|^2$ and $\EV[\|\Delta d_k\|^2]$, and is changing over time. The two extreme cases are $\rho=1$, which corresponds to the case with $\Var_d=0$, i.e., the disturbance has a deterministic profile, and $\rho=0$, which corresponds to the case when $\EV[d]\equiv0$, i.e., the disturbance has zero mean. For the case with $\rho=1$, substituting $\EV[\|\Delta d_k\|^2]=0$ into \eqref{eq:errorDiss} and taking expectation on both sides leads to the inequality
\begin{equation}
    \EV[\|e_k\|_P^2]-\EV[\|e_{k-1}\|_P^2]\leq -\EV[\|\Pi_y\F e_k\|^2],
\end{equation}
i.e., it requires the error covariance of $y_k$ to asymptotically converge to 0, which is impossible under the input uncertainty $\Delta u_k$ and the measurement noise $n_k$. It is possible to solve \eqref{eq:PD} and \eqref{eq:dissConGeneral} only, giving a performance bound of $\gamma_1$ on the expected trajectories, but it is not possible to obtain a probabilistic performance bound on the actual system trajectories. The case where $\rho=0$ reduces the problem to a special case of the next section and will be further discussed later.

Another feature of \eqref{eq:L2Prob} is that one can obtain a probabilistic $\Lt$ bound for any desired level of probability by choosing specific values of $\gamma$. In particular, we have the following corollary.
\begin{cor}\label{cor:L2ProbBound}
    For the setup in Theorem \ref{thm:dissConGeneral}, if the conditions in \eqref{eq:dissCons} are satisfied with $\gamma_1=\gamma_2=\gamma\sqrt{p}$, then the feasible controlled behavior satisfies \begin{equation}\label{eq:L2ProbBound}
			\Pb\left\{\lim_{T\rightarrow\infty}\Pb\left\{\Gamma_T\leq\gamma\right\}\geq 1-p\right\}=1.
		\end{equation}
\end{cor}
\begin{rem}\label{rem:gainsGeneralize}
    The condition for $\gamma_1$ and $\gamma_2$ in Corollary \ref{cor:L2ProbBound} is not the only way to achieve \eqref{eq:L2ProbBound}. In fact, any combinations of $\gamma_1$ and $\gamma_2$ such that 
    \begin{equation}\label{eq:L2Tradeoff}
        \rho\gamma_1^2+(1-\rho)\gamma_2^2=p\gamma^2
    \end{equation}
with $\rho$ given by \eqref{eq:SNRGeneral}, achieve the design goal. However, since the trajectories of $\EV[d]$ cannot be obtained \emph{a priori}, it is not possible to compute the value of $\rho$, and Corollary \ref{cor:L2ProbBound} gives the only condition that is independent of $\rho$. As will be discussed in the next section, more solutions are possible with additional information on $\EV[d]$.
\end{rem}

Corollary \ref{cor:L2ProbBound} provides a way to verify the feasibility of a given performance bound and a probability of success. By doing so, \eqref{eq:L2ProbBound} can be integrated into optimization-based control designs (e.g., stochastic economic data-driven predictive control) as a chance constraint. Specifically, this chance constraint can be achieved by transforming \eqref{eq:dissTransform} into the condition
\begin{equation}\label{eq:optimizationCond}
    \begin{bmatrix}
        \psi & *\\
        \F_p\hat{g}_{k-1|k-1}+\F_f\EV[d_k]+\F_z\hat{z}_k & W
    \end{bmatrix}\geq0,
\end{equation}
where
\begin{equation}
\begin{split}
    &\psi=\|\hat{g}_{k-1|k-1}\|_{M-\F^\top\Pi_y^\top\Pi_y\F}^2\\
    & \hskip 1.5cm +p\gamma^2(\|\EV[d_k]\|^2+\tr(\Var_d))\\
    & \hskip 2.5cm -\tr(M\mathcal{N}+\F^\top\Pi_y^\top\Pi_y\F\Pk).
\end{split}
\end{equation}
Note that this condition is convex in the decision variable $\hat{z}_k$.

While Theorem \ref{thm:dissConGeneral} and Corollary \ref{cor:L2ProbBound} can serve different purposes, their constructions on the expected value are mathematically similar. In particular, once a control design procedure is developed to achieve \eqref{eq:L2Prob}, a corresponding procedure is immediately available to achieve \eqref{eq:L2ProbBound} by replacing $\gamma$ by $\gamma\sqrt{p}$. As such, subsequent developments in this paper will focus on the design to achieve probabilistic finite $\Lt$-gain stability given by \eqref{eq:L2Prob}.

\subsection{Probabilistic Finite $\Lt$-Gain Stabilization - the Case with Constant Disturbance Mean}
A common case of the problem setup is when the mean of disturbance is constant, i.e., $\EV[d_k]=\bar{d}$ for all $k\in\mathbb{T}$. A practical example of this scenario is a cooling water stream in a chemical process with a temperature set at a constant value but fluctuating in reality. In such a case, this constant mean can be incorporated into the offline design stage as a known constant rather than a variable. This leads to the control design summarized by the following theorem.
\begin{thm}\label{thm:dissConConstDist}
    Let $\Sigma$ be an LTI system represented in the form \eqref{eq:fundLemma} with the manifest variable $w=(y,u,d)$, in which $d\in\mathcal{D}(\EV[d],\Var_d)$ with $\EV[d_k]=\bar{d}, \ \forall k\in\mathbb{T}$ and $u\in\mathcal{D}(\bar{u},\Var_u)$. Given real numbers $\gamma_1, \gamma_2\geq0$, if there exist matrices $W\in\Sym{\mathsf{g}}$, $X\in\Sym{\mathsf{g}}$, $Y\in\sR[\mathsf{u}]{\mathsf{g}}$ and a vector $\xi\in\sR{\mathsf{u}}$ such that \eqref{eq:PD} and \eqref{eq:offset_g} are satisfied, and
    \begin{equation}\label{eq:dissConsFixedDisturbance}
        \begin{bmatrix}
				\phi & * & *&*\\
				0&W&*&*\\
				0&\Pi_y\F W & I & *\\
				\F_f\bar{d} +\F_z\xi & \F_pW+\F_zY & 0 & W
			\end{bmatrix}\geq0,
    \end{equation}
    where
    \begin{equation}\label{eq:constMeanOffset}
        \phi=\gamma_1^2\|\bar{d}\|^2+\gamma_2^2\tr(\Var_d)-\tr(X),
    \end{equation}
    then there exists a feasible controlled behavior that satisfies \eqref{eq:L2Prob} with
    \begin{equation}\label{eq:SNRCostDist}
        \rho=\frac{\|\bar{d}\|^2}{\tr(\Var_d)+\|\bar{d}\|^2},
    \end{equation}
    and a corresponding solution of $\hat{g}_{k|k-1}$ is given by
	\begin{equation}\label{eq:controlConstructConstDist}
		\hat{g}_{k|k-1}=(\F_p+\F_zYW^{-1})\hat{g}_{k-1|k-1}+\F_z\xi,
	\end{equation}
    which leads to $\bar{u}_k$ using \eqref{eq:controlAction}.
\end{thm}
\begin{proof}
    The proof is the same as that of Theorem \ref{thm:dissConGeneral} until \eqref{eq:dissTransform} with $\EV[d_k]=\bar{d}$. Since $\bar{d}$ is constant in this case instead of a variable, \eqref{eq:dissTransform} is of the form \eqref{eq:QFbehaviorOffset} having nonzero linear terms, with $v_1=\hat{g}_{k-1|k-1}$ and $v_2=\hat{z}_k$. By Lemma \ref{lem:QFbehaviorOffset}, the virtual manipulated variable in this case should be of the form 
    \begin{equation}
        \hat{z}_k=K\hat{g}_{k-1|k-1}+\xi.
    \end{equation}
    Substituting into \eqref{eq:dissTransform} gives
     \begin{equation}
        \begin{split}
        &\|\hat{g}_{k-1|k-1}\|_{M-\F^\top\Pi_y^\top\Pi_y\F}^2\\
        & \hskip 0.4cm -\|(\F_p+\F_zK)\hat{g}_{k-1|k-1}+\F_f\bar{d}+\F_z\xi\|_M^2\\
        &\hskip 0.8cm +\gamma_1^2\|\bar{d}\|^2+ \gamma_2^2\tr(\Var_d)-\tr(M\mathcal{N}+\F^\top\Pi_y^\top\Pi_y\F\Pk)\geq0.
        \end{split}
    \end{equation}
    Similarly to \eqref{eq:slack}, the existence of a matrix $X$ satisfying \eqref{eq:offset_g} allows the above inequality to be written as
    \begin{equation}\label{eq:offsetSubstituted}
        \begin{split}
        &\|\hat{g}_{k-1|k-1}\|_{\Omega}^2-\|\Xi\hat{g}_{k-1|k-1}+\F_f\bar{d}+\F_z\xi\|_M^2+\phi\geq0,
        \end{split}
    \end{equation}
    where 
    \begin{equation}
		\Omega\coloneqq M-\F^\top\Pi_y^\top\Pi_y\F, \ \Xi\coloneqq\F_p+\F_zK.
    \end{equation}
    To ensure that $\hat{g}_{k-1|k-1}$ is free, there must exist a lower bound with respect to $\hat{g}_{k-1|k-1}$ (as a function of $\xi$) that satisfies \eqref{eq:offsetSubstituted}, which is equivalent to
	\begin{equation}\label{eq:existence}
		L\coloneqq \Omega-\Xi^\top M\Xi\geq0,\ L_\perp\Xi^\top M(\F_f\bar{d}+\F_z\xi)=0.
	\end{equation}
	The principle solution for the minimum to occur is
    \begin{equation}
        \hat{g}_{k-1|k-1}=L^\dagger\Xi^\top M(\F_f\bar{d}+\F_z\xi),
    \end{equation}
    with the minimum value
	\begin{equation}
		\begin{split}
			&\|L^\dagger\Xi^\top M(\F_f\bar{d}+\F_z\xi)\|_\Omega^2\\
            &\hskip 1cm-\|(I+\Xi L^\dagger\Xi^\top M)(\F_f\bar{d}+\F_z\xi)\|_M^2+\phi\\
			= \ &\phi-\|\F_f\bar{d}+\F_z\xi\|_M^2-\|\Xi^\top M(\F_f\bar{d}+\F_z\xi)\|_{L^\dagger}^2,
		\end{split}
	\end{equation}
	with $\phi$ given by \eqref{eq:constMeanOffset}. This minimum needs to be non-negative to ensure the implementability of the controlled behavior. With conditions in \eqref{eq:existence}, the non-negativity of the above can be transformed to
	\begin{equation}
		\begin{bmatrix}
			\phi-\|\F_f\bar{d}+\F_z\xi\|_M^2&*\\
			-\Xi^\top M(\F_f\bar{d}+\F_z\xi)&L
		\end{bmatrix}\geq0.
	\end{equation}
	Substituting in the expressions for $\Omega$, $\Xi$ and $L$, the above inequality can be rearranged to the form
	\begin{equation}
		\begin{split}
		&\begin{bmatrix}
			\phi&*\\
			0 & M
		\end{bmatrix}-[*]^\top\begin{bmatrix}
			I&0\\ 0 &M
		\end{bmatrix}
		\begin{bmatrix}
			0 & \Pi_y\F \\
			\F_f\bar{d}+\F_z\xi & \F_p+\F_zK
		\end{bmatrix}\geq0.
		\end{split}
	\end{equation}
    Taking the Schur complement and applying the congruence transformation with $\diag(1, W, I, I)$ lead to \eqref{eq:dissConsFixedDisturbance}. Lastly, since $\sum_{k=1}^T\|\EV[d_k]\|^2=N\|\bar{d}\|^2$, the parameter $\rho$ in \eqref{eq:SNRGeneral} specializes to \eqref{eq:SNRCostDist}.
\end{proof}
\begin{rem}
    While $\bar{d}$ does not appear in \eqref{eq:controlConstructConstDist}, its information is encoded in $\xi$ through \eqref{eq:dissConsFixedDisturbance} and therefore used implicitly.
\end{rem}
\begin{rem} \label{rem: opm gamma}
In the case of constant disturbance mean, the parameter $\rho$ is independent of the time horizon $N$ and can be computed \emph{a priori}. As such, all solutions satisfying \eqref{eq:L2Tradeoff} can be obtained, and the performance requirement becomes a trade-off between $\gamma_1$ and $\gamma_2$ during offline design. Since the primary goal is to optimize the performance in \eqref{eq:L2Prob}, $\rho\gamma_1^2+(1-\rho)\gamma_2^2$ can be used as an objective function to be minimized during offline design. As shown by the relationship in Remark \ref{rem:gainsGeneralize}, this will lead to the lowest probability of failure $p$ given any performance requirement $\gamma$ (or the lowest value of $\gamma$ for any given $p$ if the goal is to achieve the best possible performance for a given confidence level). 
\end{rem}
\begin{rem}
While the design is for constant disturbance mean, this approach also applies to the case when the disturbance mean profile is piece-wise constant with values within a finite set. In such a case, a set of controllers can be designed offline and implemented according to the online value.
\end{rem}

Compared with the result of the general case given by Theorem \ref{thm:dissConGeneral}, one of the most distinct differences in Theorem \ref{thm:dissConConstDist} is that $\gamma_1$ and $\gamma_2$ are in the same inequality rather than in two different ones. This is because the constant $\bar{d}$ in this case can be treated as an \emph{a priori} known constant rather than a variable. As such, the design of $\gamma_1$ is specific to $\bar{d}$ rather than to all $\EV[d_k]$, making it a component of $\beta$ rather than $v_1$ in \eqref{eq:QFbehaviorOffset}. This also means that $\gamma_2$ does not need to accommodate the worst-case $\Lt$ gain among all trajectories of $\EV[d]$, making it possible to achieve a smaller value of $\gamma_1$ compared with that in Theorem \ref{thm:dissConGeneral}.  

A special case of this scenario is when $\bar{d}=0$, which corresponds to the case when $\rho=0$. In such a case, the choice of $\gamma_1$ is no longer meaningful, and the design on the fiber (i.e., controlled system considering $\EV[y]$) becomes a stabilization problem. However, for the actual trajectories, it is still possible to characterize its probabilistic performance. This leads to the following corollary.
\begin{cor}
    For the setup in Theorem \ref{thm:dissConConstDist} with $\bar{d}=0$, if, for $\gamma_2\geq0$, there exist matrices $W\in\Sym{\mathsf{g}}$, $X\in\Sym{\mathsf{g}}$ and $Y\in\sR[\mathsf{u}]{\mathsf{g}}$ such that \eqref{eq:PD} -- \eqref{eq:offset_g} are satisfied, and
    \begin{equation}\label{eq:dissConsZeroDisturbance}
        \begin{bmatrix}
				W&*&*\\
				\Pi_y\F W & I & *\\
				\F_pW+\F_zY & 0 & W
			\end{bmatrix}\geq0,
    \end{equation}
    then there exists a feasible controlled behavior with $\EV[y]$ being asymptotically stable, i.e., $\lim_{k\rightarrow\infty}\|\EV[y_k]\|=0$, and
    \begin{equation}\label{eq:L2ProbZeroDist}
			\Pb\left\{\lim_{T\rightarrow\infty}\Pb\left\{\Gamma_T\leq\gamma\right\}\geq 1-\frac{\gamma_2^2}{\gamma^2}\right\}=1.
		\end{equation}
    A corresponding implementing value of $\bar{u}_k$ can be obtained from \eqref{eq:controlAction} with 
    \begin{equation}\label{eq:controlConstructZeroDist}
		\hat{g}_{k|k-1}=(\F_p+\F_zYW^{-1})\hat{g}_{k-1|k-1}.
	\end{equation}
\end{cor}
\begin{proof}
    Since $\bar{d}=0$, \eqref{eq:dissConsFixedDisturbance} has solutions if and only if solutions exist when $\xi=0$. In such a case, \eqref{eq:dissConsFixedDisturbance} is equivalent to the combination of \eqref{eq:offsetGeneral} and \eqref{eq:dissConsZeroDisturbance}. In other words, \eqref{eq:controlConstructZeroDist} implements a controlled behavior satisfying \eqref{eq:L2Prob}. Furthermore, $\bar{d}=0$ implies that $\rho=0$, which specializes \eqref{eq:L2Prob} to \eqref{eq:L2ProbZeroDist}.

    The asymptotic stability of $\EV[y]$, can be shown directly by setting $\EV[d_k]$ in \eqref{eq:EVDiss} to 0, making $\|\EV[g]\|_P^2$ a Lyapunov function for the controlled expected behavior.
\end{proof}
\subsection{The Complete Design and Implementation Procedure}
The offline design and online implementation procedure resulting from the above results are summarized in Procedure \ref{proc:DDPC}. The procedure treats $\gamma_1$ and $\gamma_2$ as given values, but they can also be decision variables to be optimized while solving \eqref{eq:dissCons} or \eqref{eq:dissConsFixedDisturbance}, as discussed in the previous sections. During online implementation, if the control is optimization-based, then the computation of $\hat{g}_{k|k-1}$ in lines \ref{step:onlinePriorStart} -- \ref{step:onlinePriorEnd} should use \eqref{eq:optimizationCond} instead of \eqref{eq:controlConstruct} (or \eqref{eq:controlConstructConstDist} in the case of constant disturbance mean) to decide $\hat{z}_k$ first, followed by using \eqref{eq:prior}. As discussed in Section \ref{subsec:dissGeneral}, while the offline design uses the solution to the ARE \eqref{eq:ARE} as the estimation error covariance, the online implementation procedure still uses the dynamic optimal filtering algorithm described in Section \ref{subsec:estimation} to obtain the actual estimation error covariance. 

\begin{figure}[!tb]
	\removelatexerror
	\begin{algorithm}[H]
		\caption{Probabilistic Finite $\Lt$-Gain Stabilization Design and Implementation}
		\label{proc:DDPC}
		\SetKwInput{Input}{Input}
		\SetKwInput{Output}{Output}
		\Input{$\mathcal{W}$, $\gamma_1$, $\gamma_2$, $T$}
		Offline Design\label{step:step 1}\;
		\begin{alglist}
                Construct the matrix $\F$ using the measured trajectories in $\mathcal{W}$ following Lemma \ref{lem:behaviorApprox}.\;
                Compute $\F_p$, $\F_f$ and $\F_z$ using \eqref{eq:gIter}.\;
                Compute $\E_p$, $\E_f$ and $\E_u$ using \eqref{eq:errorCoeff}.\;
			Obtain matrix $\Pk$, the solution to the ARE \eqref{eq:ARE}.\;
			\eIf{$\EV[d_k]$ is not unchanged for all $k$}{
				Solve \eqref{eq:dissCons} with parameters $\gamma_1$ and $\gamma_2$.\;}{
				Solve \eqref{eq:dissConsFixedDisturbance} with parameters $\gamma_1$ and $\gamma_2$.\;}	
		\end{alglist}
		Online Implementation \label{step:step 2} \;
		\begin{alglist}
            At $k=1$, initialize $\hat{g}_{0|0}$ and $\Pk_{0|0}$.\;
            \While{$k\leq T$}{\label{step:onlinestart}
			    \eIf{$\EV[d_k]$ is not unchanged for all $k$}{\label{step:onlinePriorStart}
				Compute $\hat{g}_{k|k-1}$ as \eqref{eq:controlConstruct}.\;}{
				Compute $\hat{g}_{k|k-1}$ as \eqref{eq:controlConstructConstDist}.\;}\label{step:onlinePriorEnd}
                Implement $\bar{u}_k$ following \eqref{eq:controlAction}.\;
                Obtain $\Pk_{k|k-1}$ following \eqref{eq:KFIterStart}.\;
                Calculate $\hat{w}_{k|k-1}$ as $\hat{w}_{k|k-1}=\Pi_f\F\hat{g}_{k|k-1}$.\;
                Obtain $\mathcal{K}_k$ and $\Pk_{k|k}$ using \eqref{eq:KFGain} and \eqref{eq:KFIterEnd}, respectively.\;
                Measure $w_k^m$ and compute $\hat{g}_{k|k}$ using \eqref{eq:posterior}.\;
			Set $k\gets k+1$ and return to Step \ref{step:onlinestart}.\;
            }
		\end{alglist}
	\end{algorithm}
\end{figure}

\section{Numerical Example}\label{sec:example}
Consider an LTI system with kernel representation
\begin{equation}\label{eq:exampleKernel}
    R_y(\sigma^{-1})y_k+R_u(\sigma^{-1})u_k+R_d(\sigma^{-1})d_k=0,
\end{equation}
where
\begin{align*}
    R_y(\sigma^{-1})&=\begin{bmatrix}
        4.29-4.5\sigma^{-1} & -1.43+1.5\sigma^{-1}\\
        -1.43-1.57\sigma^{-1} & 2.14+2.36\sigma^{-1}
        \end{bmatrix}\\
    R_u(\sigma^{-1})&=\begin{bmatrix}
       -1.11-0.65\sigma^{-1} & -1.4+0.42\sigma^{-1}\\
        -1.47+0.024\sigma^{-1} & -1.45-0.17\sigma^{-1}
        \end{bmatrix}\\
    R_d(\sigma^{-1})&=\begin{bmatrix}
            -0.15 & -0.12\\
             -0.11 & -0.16
            \end{bmatrix}
\end{align*}
Using a discrete-time counterpart of the criteria in \cite[Theorem 7.6.2]{Polderman:1998}, it can be verified that this system is unstable.

The covariance matrices for $\Delta d$, $\Delta u$ and $n$ are $\Var_d= \diag(0.4,0.35),
\Var_u= \diag(0.2,0.1)$ and $\Var_n= \diag(0.6,0.2,0.1,0.5,0.5,0.3)$, respectively. In this example, the representation \eqref{eq:exampleKernel} was only used for data generation. We set $L = 4$ and generated trajectories of $\Delta d_k$, $\Delta u_k$ and $n_k$ using Gaussian mixture models. The example consists of the general case and the case with a constant disturbance mean. 

\subsection{The General Case}
For the general case, simulation studies were carried out by choosing $\gamma_1^2=\gamma_2^2=0.81$. Fig~\ref{fig:iotraj-gen} shows the control performance of 50 different trajectories in the controlled system under different disturbances. It can be seen that all trajectories have been stabilized when $\mathbb{E}[d_k]$ is time-varying. To show the probabilistic finite $\Lt$ gain performance, the cumulative distribution function of $\Gamma_T$ in~\eqref{eq:truncL2} for $T = 5, 20$ and $100$ are shown in Fig~\ref{fig:cdf-gen-single}, with the black dashed line showing the bound of the inner probability. Due to the combined effect of initial estimation error, the stochasticity of $\Delta d$ and the mismatch between $\Pk$ and $\Pk_{k|k}$, the inner probability bound is unlikely to be satisfied, as is shown by the red curve for $T=5$. However, as $T$ increases, the estimation error covariance converges and strong law of large numbers is gradually satisfied, the cumulative density gradually moves to intersect with the bound (as shown by the blue curve when $T=20$) and is eventually completely above the bound, as shown by the magenta curve when $T=100$. Furthermore, the above experiment has been repeated 50 times, each with a different distribution of $\Delta d$ but with the same covariance $\Var_d$. The cumulative distribution curves are depicted in Fig~\ref{fig:cdf-gen}, which shows that all cumulative distribution functions of $\Gamma_{100}$ are above the dashed curve. While one would need infinitely many trajectories to fully validate the proposed approach, the sufficiently large sample base gives a strong and promising indication of its effectiveness.

\begin{figure}[tb]
    \centering
    \includegraphics[width=\linewidth]{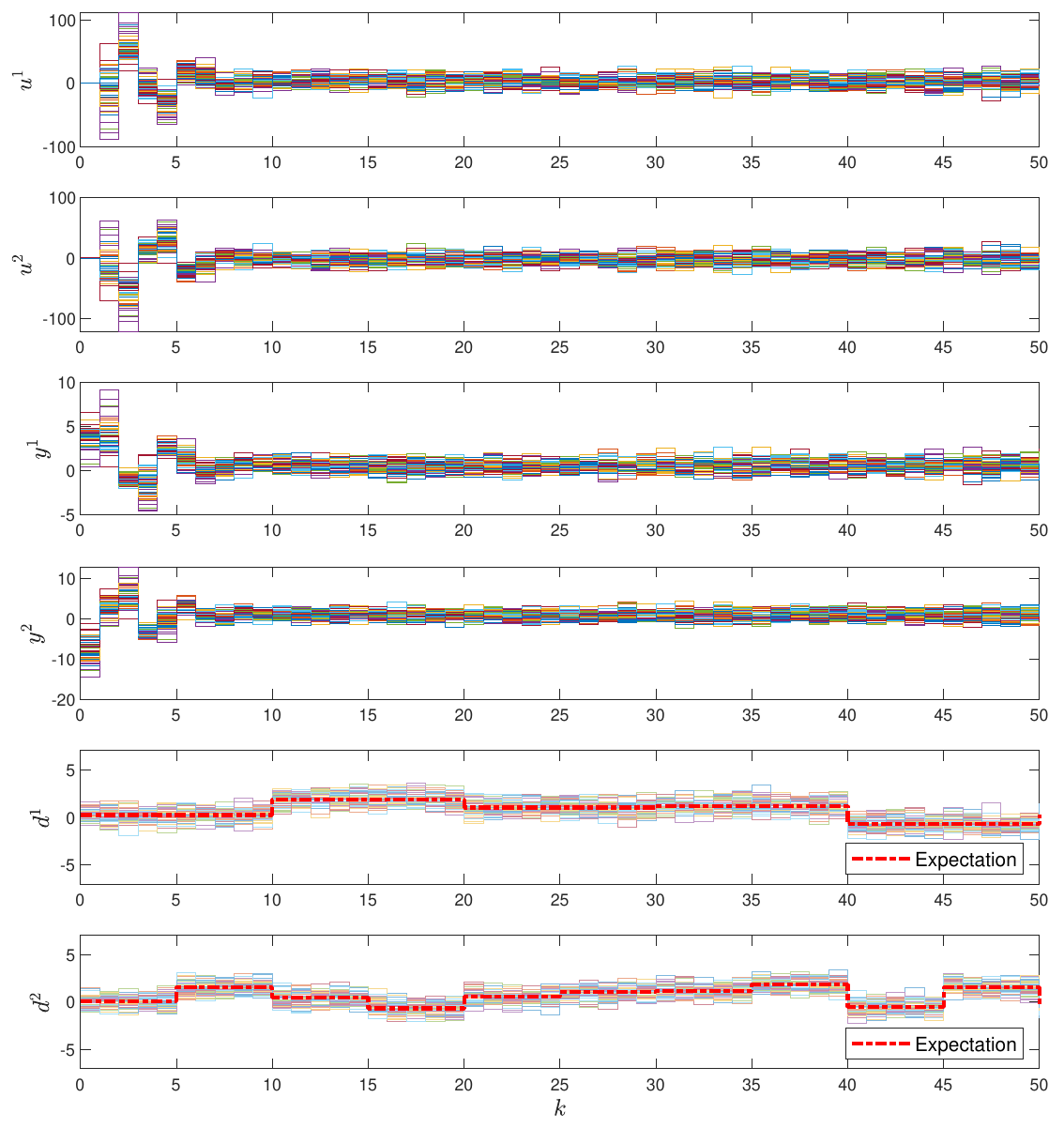}
    \caption{ $\mathcal{L}_2$ stabilization performance in the general case ($\gamma_1^2 = \gamma_2^2 = 0.81$): fifty different closed-loop input-output and disturbance trajectories and time-varying $\mathbb{E}[d_k]$.}
    \label{fig:iotraj-gen}
\end{figure}

\begin{figure}[tb]
    \centering
    \includegraphics[width=\linewidth]{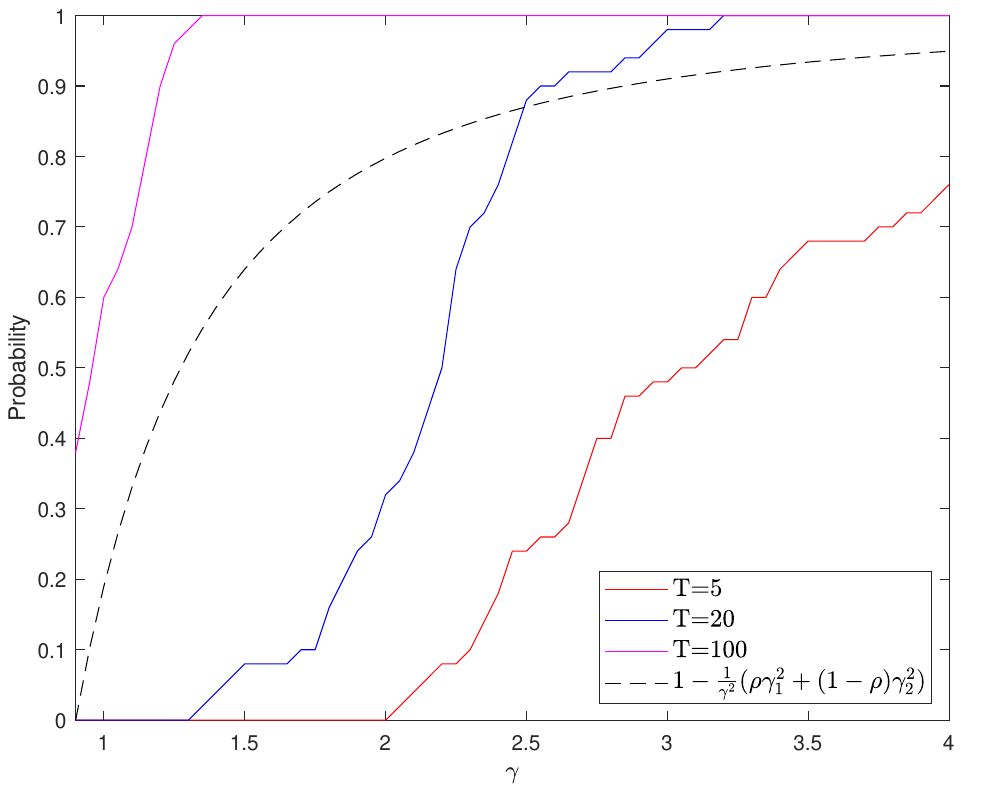}
    \caption{The inner probability test in the general case ($\gamma_1^2 = \gamma_2^2 = 0.81$): cumulative distribution functions of $\Gamma_T$ for different values of $T$.}
    \label{fig:cdf-gen-single}
\end{figure}

\begin{figure}[tb]
    \centering
    \includegraphics[width=\linewidth]{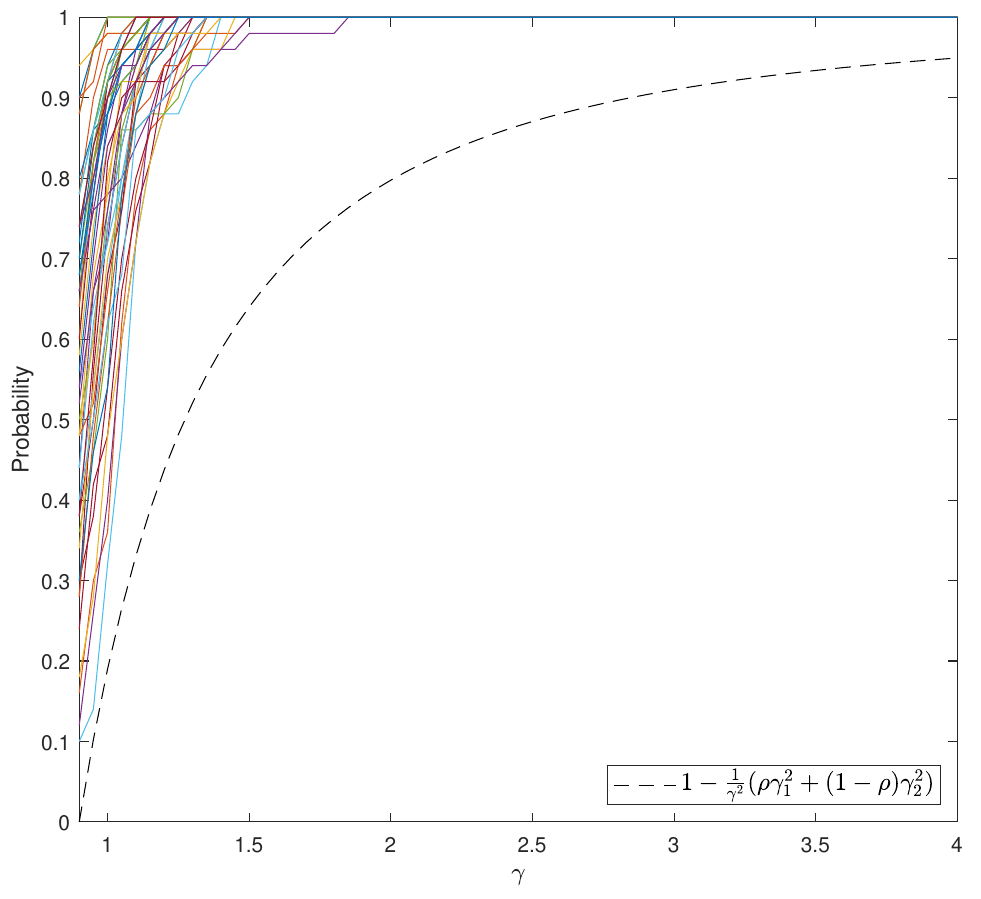}
    \caption{The outer probability test in the general case ($\gamma_1^2 = \gamma_2^2 = 0.81$): fifty different cumulative distribution functions of $\Gamma_T$ for $T = 100$. }
    \label{fig:cdf-gen}
\end{figure}

\subsection{The case with a Constant Disturbance Mean}
For the case with constant mean, we begin by comparing its performance with the general case using the same choice of $\gamma_1$ and $\gamma_2$. The inner probability test and the outer probability test  are shown in Fig~\ref{fig:const-inner-same-gamma} and Fig~\ref{fig:const-outer-same-gamma}, respectively. Comparing  Fig~\ref{fig:const-inner-same-gamma} with Fig~\ref{fig:cdf-gen-single}, the inner probability and the outer probability can be satisfied for all $\gamma$ with a smaller number of steps $T$ in the case with a constant disturbance mean. Furthermore, comparing Fig~\ref{fig:const-outer-same-gamma} with Fig~\ref{fig:cdf-gen} shows that the additional information of constant disturbance yields a better performance for the outer probability.

Simulation has also been carried out with $\gamma_1$ and $\gamma_2$ determined by optimizing $\rho\gamma_1^2+(1-\rho)\gamma_2^2$ while solving the LMIs in Theorem \ref{thm:dissConConstDist}. The values of $\gamma_1$ and $\gamma_2$ have been obtained as $\gamma_1^2$ = 0.22 and $\gamma_2^2$ = 0.36. The cumulative distributions for different values of $T$ have been depicted in Fig~\ref{fig:const-inner-diff-gamma}, which shows that the inner probability bound can still be satisfied eventually, but the time to achieve the bound is longer. Furthermore, Fig~\ref{fig:const-outer-diff-gamma} shows that the optimized values of $\gamma_1$ and $\gamma_2$ provide a much tighter bound for the cumulative distributions, especially when the value of $\gamma$ is small.

\begin{figure}[tb]
    \centering
    \includegraphics[width=\linewidth]{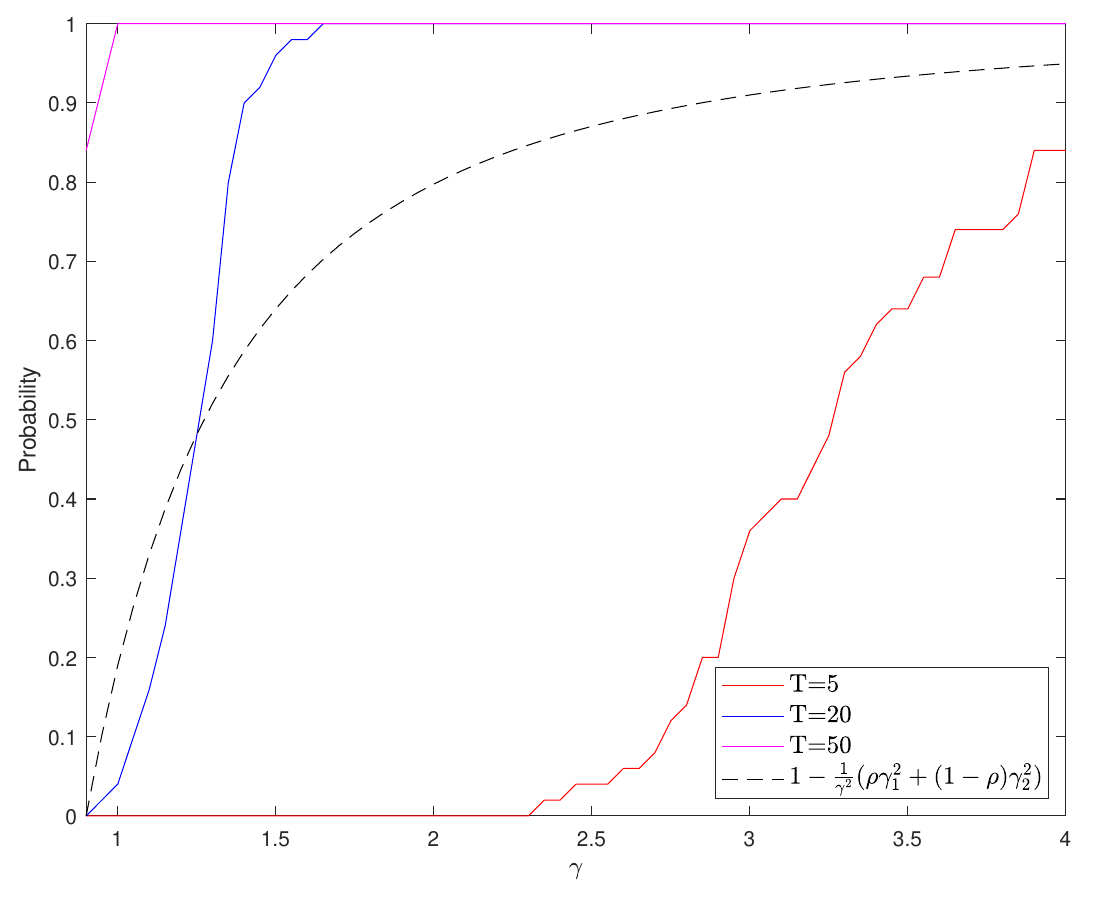}
    \caption{The inner probability test in the case with a constant disturbance mean ($\gamma_1^2 = \gamma_2^2 = 0.81$): cumulative distribution functions of $\Gamma_T$ for different values of $T$.}
    \label{fig:const-inner-same-gamma}
\end{figure}

\begin{figure}[tb]
    \centering
    \includegraphics[width=\linewidth]{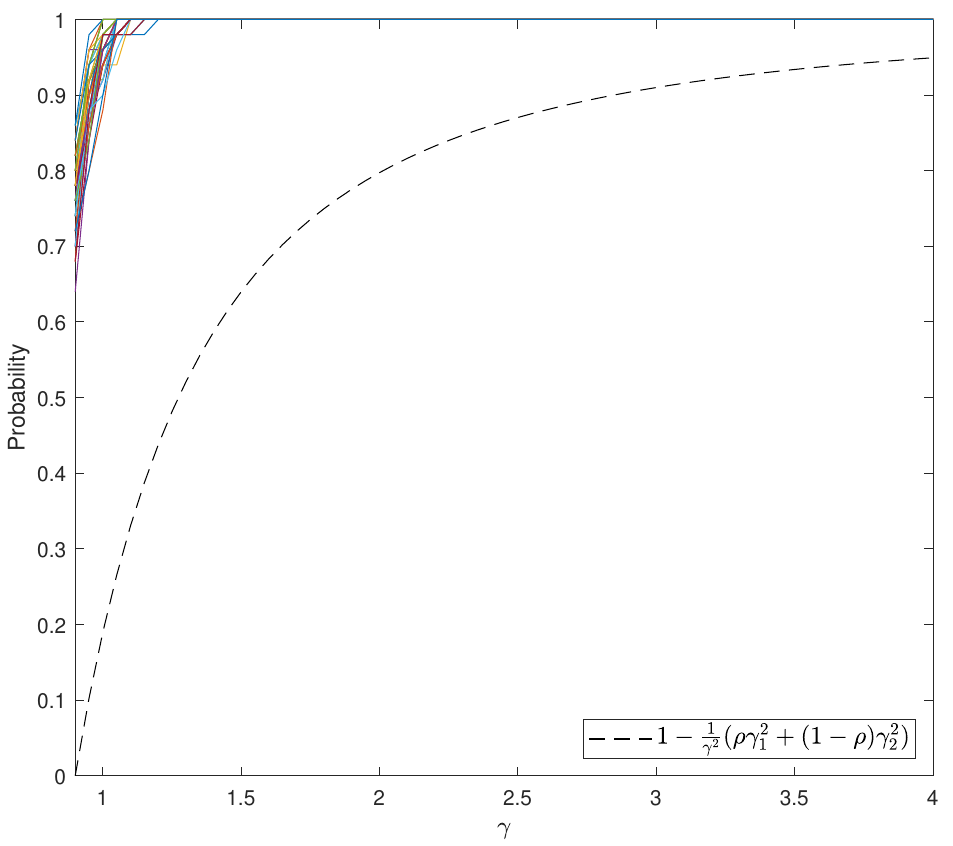}
    \caption{The outer probability test in the case with a constant disturbance mean ($\gamma_1^2 = \gamma_2^2 = 0.81$): fifty different cumulative distribution functions of $\Gamma_T$ for $T = 50$. }
    \label{fig:const-outer-same-gamma}
\end{figure}

\begin{figure}[tb]
    \centering
    \includegraphics[width=\linewidth]{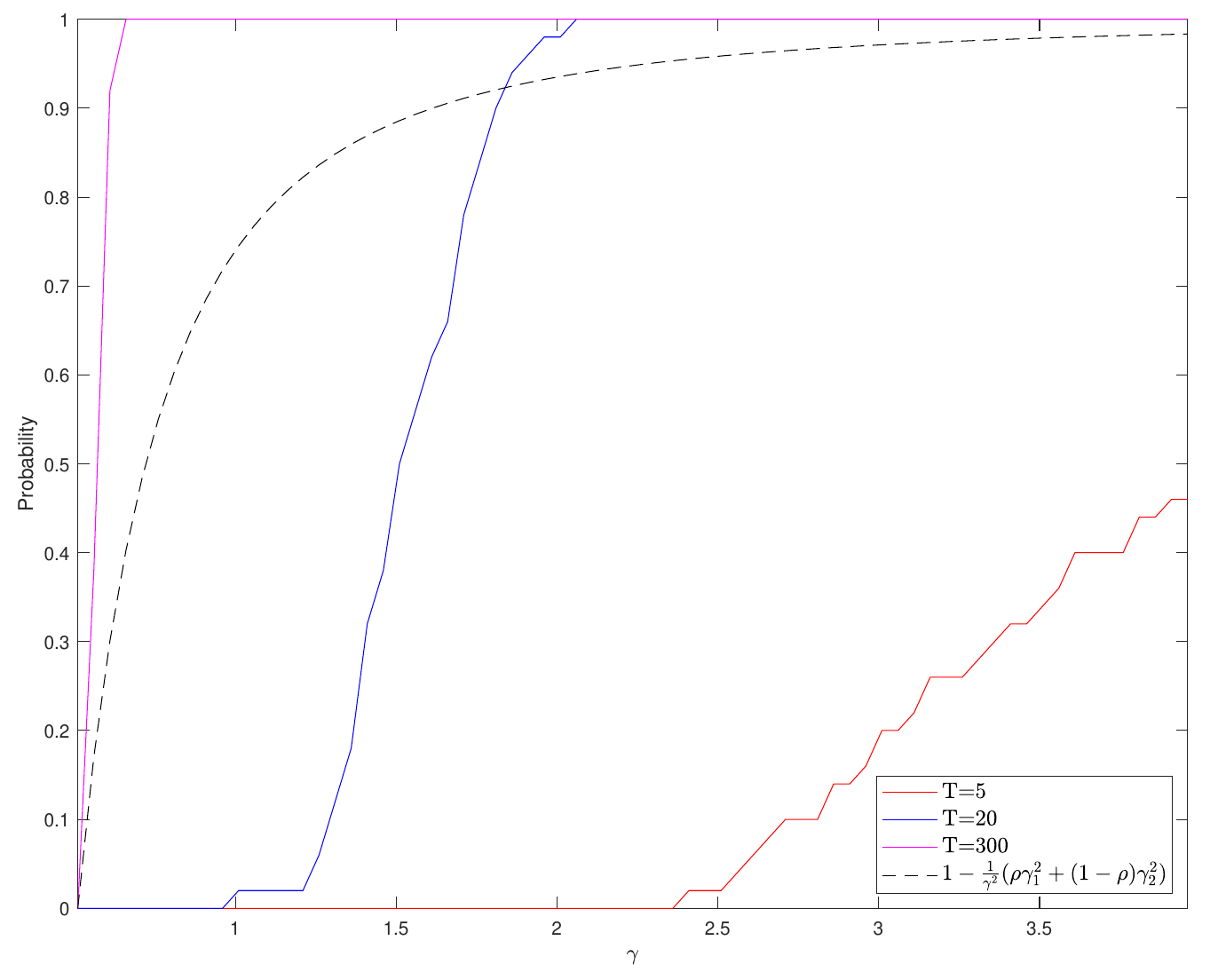}
    \caption{The inner probability test in the case with a constant disturbance mean ($\gamma_1^2 = 0.22$ and $\gamma_2^2 = 0.36$): cumulative distribution functions of $\Gamma_T$ for different values of $T$.}
    \label{fig:const-inner-diff-gamma}
\end{figure}

\begin{figure}[tb]
    \centering
    \includegraphics[width=\linewidth]{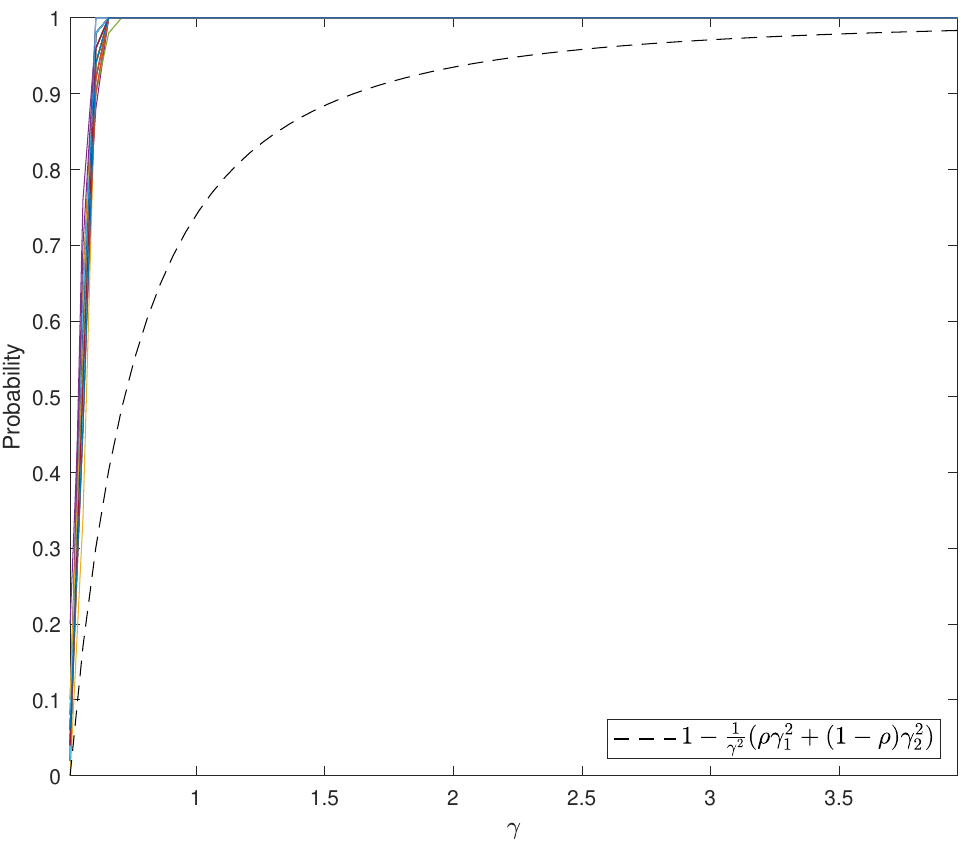}
    \caption{The outer probability test in the case with a constant disturbance mean ($\gamma_1^2 = 0.22$ and $\gamma_2^2 = 0.36$): fifty different cumulative distribution functions of $\Gamma_T$ for $T = 300$. }
    \label{fig:const-outer-diff-gamma}
\end{figure}

\section{Conclusion}\label{sec:conclusion}
In this paper, we have developed a new probabilistic framework for data-driven finite $\Lt$-gain stabilization control of stochastic LTI systems. Leveraging on the state property of the parameterizer, an optimal filtering algorithm is developed to estimate the true trajectory from noisy measurements, and the estimation error covariance is directly integrated into the offline control design, leading to a control algorithm that characterizes the probability of success for any given performance level.

This probabilistic design framework opens up many possibilities for new research directions. An immediate future work is to extend this approach to distributed control of interconnected systems. The uncertainties for each subsystem propagate in the interconnection network and indirectly affect other subsystems, which means that the estimation algorithm also needs to be extended to a distributed version. Furthermore, as shown in the simulation studies, the fact that Theorems \ref{thm:dissConGeneral} and \ref{thm:dissConConstDist} cater for all possible distributions of $\Delta d$ inevitably makes the probability bound loose for many distributions. In the event when more statistical information of $\Delta d$ is known, it would be interesting to investigate whether it is possible to obtain a tighter bound, yielding less conservative control design.

\appendix

\subsection{Proof of Lemma \ref{lem:QFbehaviorOffset}}\label{appx:proofQFbehaviorOffset}
The \emph{if} part is obvious. For the \emph{only if} part, consider first the extended variable $(v_0,v_1,v_2)$ satisfying
        \begin{equation}
            \begin{bmatrix}\label{eq:QFBehaviorExtend}
			    v_0\\ v_1\\ v_2
			\end{bmatrix}^\top\begin{bmatrix}
                \beta & \frac{1}{2}\eta^\top & \frac{1}{2}\mu^\top\\
                \frac{1}{2}\eta &Q & S\\ 
                \frac{1}{2}\mu & S^\top & -R
            \end{bmatrix} \begin{bmatrix}
			    v_0\\ v_1\\ v_2
			\end{bmatrix}\geq0.
        \end{equation}
Using \cite[Lemma 7]{Yan:2025}, a necessary and sufficient condition for both $v_0$ and $v_1$ to be free is the existence of $\tau\in\mathbb{R}$ such that
\begin{equation}\label{eq:extendFree}
            \begin{bmatrix}
                \beta & \frac{1}{2}\eta^\top\\
                \frac{1}{2}\eta & Q
            \end{bmatrix}+\begin{bmatrix}
                \frac{1}{2}\mu^\top \\ S
            \end{bmatrix}(R^\dagger+\tau R_\perp)\begin{bmatrix}
                \frac{1}{2}\mu & S^\top
            \end{bmatrix}\geq0,
        \end{equation}
and a solution of $v_2$ for any given $v_0$ and $v_1$ can be constructed as
\begin{equation}\label{eq:QfSol}
    v_2=\left(R^\dagger+\frac{\tau}{2}R_\perp\right)\begin{bmatrix}
                \frac{1}{2}\mu & S^\top
            \end{bmatrix}\begin{bmatrix}
                v_0\\ v_1
            \end{bmatrix}.
\end{equation}
Now, comparing \eqref{eq:QFBehaviorExtend} with \eqref{eq:QFbehaviorOffset} shows that the latter is a special case of the former with $v_0\equiv 1$. As such, we can substitute $v_0=1$ to \eqref{eq:QfSol}, giving
\begin{equation}
    v_2=\underbrace{\left(R^\dagger+\frac{\tau}{2}R_\perp\right)S^\top}_Kv_1+\underbrace{\frac{1}{2}\left(R^\dagger+\frac{\tau}{2}R_\perp\right)\mu}_\xi,
\end{equation}
which is of the form \eqref{eq:gainQF}.

        In the case when $\mu=0$, \eqref{eq:extendFree} can be written as 
        \begin{equation}
            \begin{bmatrix}
                \beta & \frac{1}{2}\eta^\top\\
                \frac{1}{2}\eta & Q+S(R^\dagger+\tau R_\perp)S^\top
            \end{bmatrix}\geq0,
        \end{equation}
which immediately implies that $\beta\geq0$. In such a case, a possible solution of $v_2$ can be derived from \eqref{eq:QfSol} as
        \begin{equation}
            v_2=\left(R^\dagger+\frac{\tau}{2}R_\perp\right)S^\top v_1,
        \end{equation}
        i.e., $\xi$ in \eqref{eq:gainQF} can be chosen to be 0.

\bibliographystyle{IEEEtran}

\bibliography{refUpdated}

\end{document}